\documentclass[12pt]{article}
\usepackage{enumerate}

\usepackage{amsmath,amsfonts,amsthm,amssymb,graphicx,natbib,geometry,arydshln,umoline,appendix,subfig,newtxtext,newtxmath,comment}

\usepackage{color, comment}

\newtheorem{lemma}{Lemma}

\newtheorem{proposition}{Proposition}
\newtheorem{corollary}[proposition]{Corollary}

\theoremstyle{definition}
\newtheorem{definition}{Definition}

\theoremstyle{remark}
\newtheorem{remark}{Remark}

\def\var{{\mathrm{var}}}
\def\cov{{\mathrm{cov}}}
\def\E{{\mathrm{E}}}
\def\I{{\mathrm{I}\,}}

\def\vart{{\mathrm{var}_{\tilde\theta}}}
\def\covt{{\mathrm{cov}_{\tilde\theta}}}
\def\Et{{\mathrm{E}_{\tilde\theta}}}

\def\I{{\mathrm{I}\,}}
\def\It{{\mathrm{I}_{\tilde\theta}\,}}

\title{Impacts of Public Information\\on Flexible Information Acquisition\thanks{I thank seminar participants for valuable discussions and comments at Otaru University of Commerce. I acknowledge financial support by MEXT, Grant-in-Aid for Scientific Research (18H05217).}}

\author{Takashi Ui
\\Hitotsubashi University\\\texttt{oui@econ.hit-u.ac.jp}}

\date{April 2022}

\begin{document}

\maketitle

\begin{abstract}
Interacting agents receive public information at no cost and flexibly acquire private information at a cost proportional to entropy reduction. When a policymaker provides more public information, agents acquire less private information, thus lowering information costs. Does more public information raise or reduce uncertainty faced by agents? Is it beneficial or detrimental to welfare? To address these questions, we examine the impacts of public information on flexible information acquisition in a linear-quadratic-Gaussian game with arbitrary quadratic material welfare. More public information raises uncertainty if and only if the game exhibits strategic complementarity, which can be harmful to welfare. However, when agents acquire a large amount of information, more provision of public information increases welfare through a substantial reduction in the cost of information. We give a necessary and sufficient condition for welfare to increase with public information and identify optimal public information disclosure, which is either full or partial disclosure depending upon the welfare function and the slope of the best response.    

\end{abstract}

\newpage

\section{Introduction}

Consider a policymaker (such as a central bank) who discloses public information to agents (such as firms and consumers). 
Agents receive public information provided by the policymaker at no cost and flexibly acquire private information at a cost proportional to entropy reduction \citep{sims2003,yang2015,denti2020}. 
When the policymaker provides more public information, agents acquire less private information, thus lowering information costs. 
Then, does the total amount of information about an unknown state increase or decrease? 
Should the policymaker provide more information to increase welfare? 
What is the optimal disclosure of public information maximizing welfare?

We address these questions in a two-stage game based on a symmetric linear-quadratic-Gaussian (LQG) game with a continuum of agents, where a payoff function is quadratic and an underlying state is normally distributed. 
In period 1, the policymaker chooses the precision of a normally distributed public signal. 
The policymaker's objective function is social welfare, which is a material benefit minus the agents' information cost, and the material benefit is an arbitrary quadratic function such as the aggregate payoff. 
In period~2, each agent flexibly acquires private information in the LQG game. 
The second-period subgame is essentially the same as those studied by \citet{denti2020} and \cite{rigos2020}.  

We measure the total amount of information in terms of the reduction in entropy caused by receiving public and private signals, i.e., the mutual information. 
By the chain rule of information, it can be represented as the sum of the mutual information of a public signal and a state and the conditional mutual information of a private signal and a state. 
When the policymaker provides more precise public information, the former increases, but the latter decreases. 
This effect of public information is referred to as the crowding-out effect \citep{colombofemminis2014}.\footnote{This effect is also found in \citet{colombofemminis2008}, \citet{wong2008}, \citet{hellwigveldkamp2009}, \citet{myattwallace2012}, among others.}  

Our first main result shows that more provision of public information decreases 
the total amount of information if and only if the game exhibits strategic complementarity (i.e., the slope of the best response is strictly positive). 
This is because the crowding-out effect is more significant under strategic complementarity than under strategic substitutability \citep{colombofemminis2014}. 
Consider the benchmark case of no strategic interaction (i.e., the slope equals zero). 
In this case,  public and private signals are perfect substitutes because an agent pays no attention to the opponents' actions. Thus, when the policymaker increases one unit of public information, an agent reduces the same unit of private information, so the total amount of information remains constant. 
In a game with strategic complementarity, however, public and private signals are imperfect substitutes and each agent puts more value on public information than on private information from a coordination motive. 
Consequently, when the policymaker increases one unit of more valuable public information, an agent reduces more units of less valuable private information, so the total amount of information decreases.

To study the welfare implication of the above result, we represent the expected welfare (in the equilibrium of the second-period subgame) as a linear combination of the volatility (the variance of the average action) and the dispersion (the variance of individual actions around the average action) minus the cost of information. This follows \citet{uiyoshizawa2015} who study the social value of public information in the case of exogenous private information.

\begin{figure} 
  \centering
  \subfloat[$\alpha< 1/2$. The slope of the best response is small.]{\label{fig:case 1}
  \includegraphics[width=3.55cm, bb=0 0 450 464]{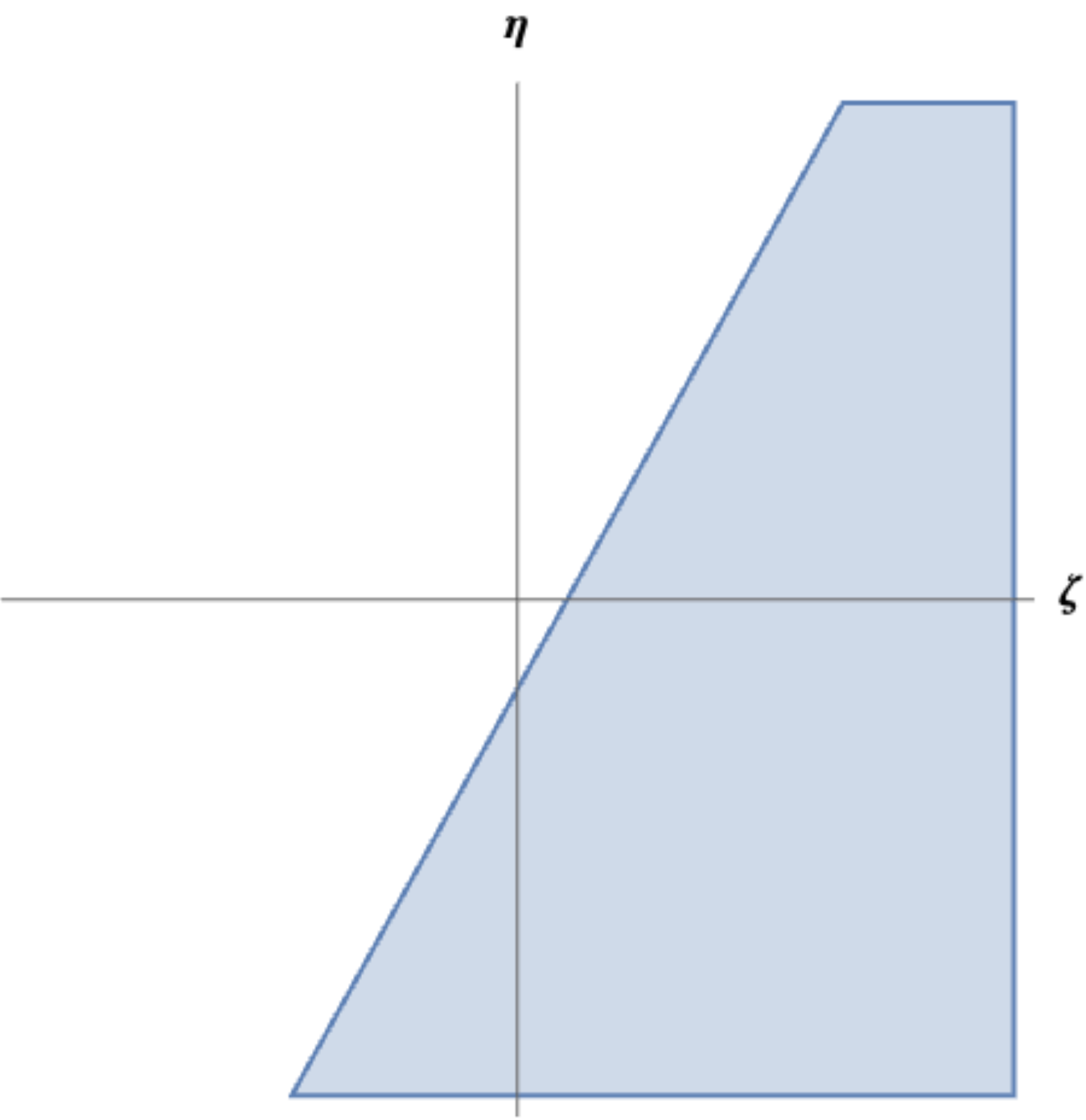}}\qquad \ 
  \subfloat[$\alpha>1/2$. The slope of the best response is large.]{\label{fig:case 2}\includegraphics[width=3.55cm, bb=0 0 450 464]{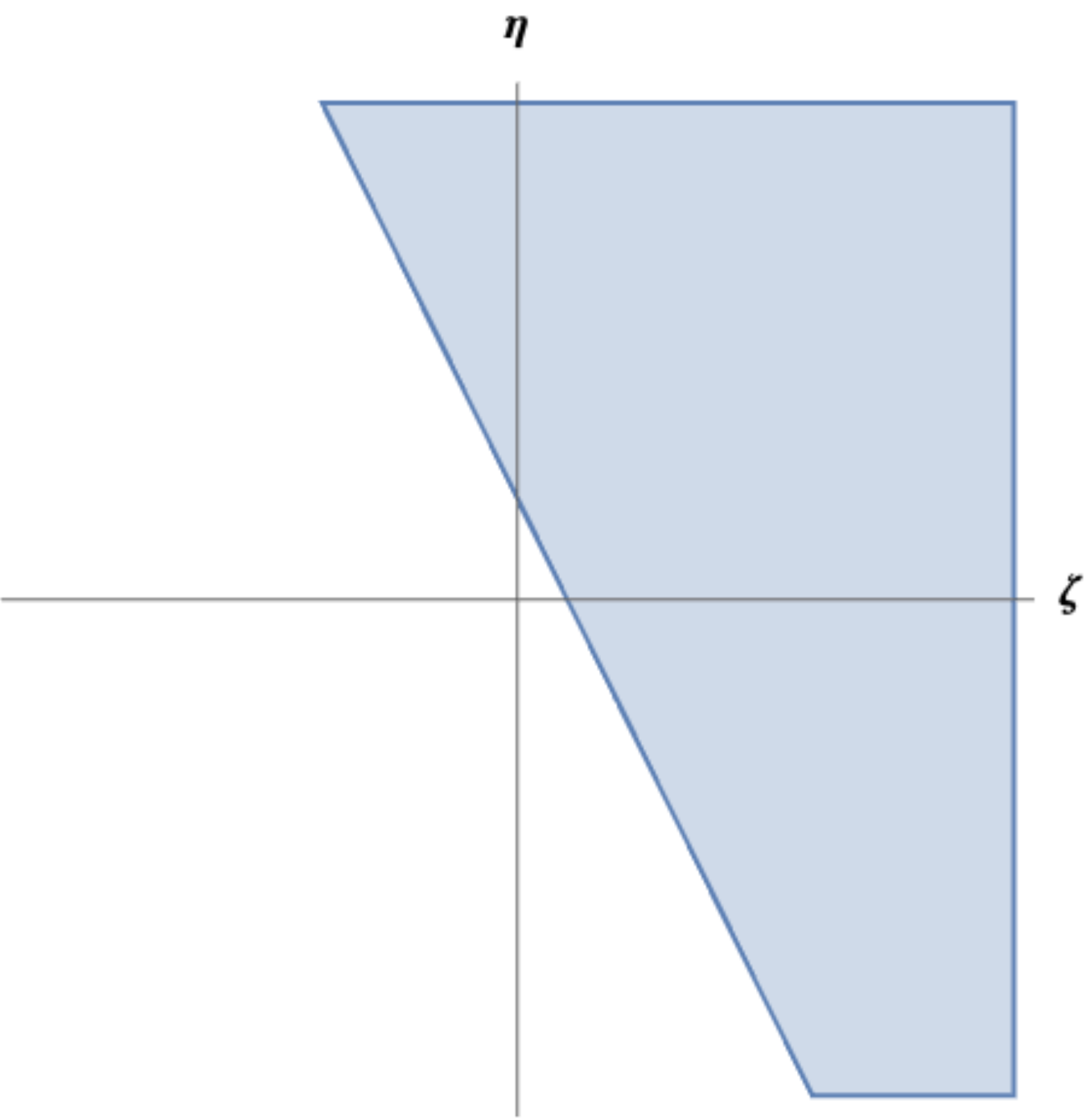}}
  \caption{Public information can be detrimental in the right region on the $\xi\eta$-plane defined by $\zeta -(1-2 \alpha ) \eta /(1-\alpha)^2>1$.} 
  \label{fig1}
\end{figure}

Our second main result is a necessary and sufficient condition for welfare to increase with public information, which depends upon the coefficients of volatility and dispersion and the slope of the best response. 
When the precision of public information is very low, 
more precise public information always increases welfare. This is because agents acquire a substantial amount of information facing considerable uncertainty, incurring an inefficiently large cost.  
When the precision of public information is not so low, more precise public information can reduce welfare if the coefficient of dispersion is greater than a threshold value determined by the coefficient of volatility (see Figure \ref{fig1}, where $\zeta$ is the coefficient of dispersion, $\eta$ is that of volatility, and $\alpha$ is the slope of the best response). 
In addition, the threshold value is increasing in the coefficient of volatility if the slope of the best response is not so large, but it is decreasing otherwise. 
This result is attributed to the following properties of dispersion and volatility. 
The dispersion decreases with public information because it equals the difference between the variance and the covariance of individual actions, and more precise public information brings them closer. 
In contrast, if the slope of the best response is not so large, the volatility increases with public information because it equals the covariance of individual actions, and more precise public information increases the covariance through an increase in the correlation coefficient. 
These properties are essentially the same as in the case of exogenous private information \citep{uiyoshizawa2015} or rigid information acquisition \citep{ui2022}. 
However, if the slope of the best response is very large, there is a sharp difference. That is, the volatility decreases with public information due to the crowding-out effect enhanced by strong strategic complementarity.

\begin{figure} 
  \centering
  \subfloat[Flexible information acquisitoin]{\label{fig:flexible}
    \includegraphics[width=3.55cm, bb=0 0 675 698]{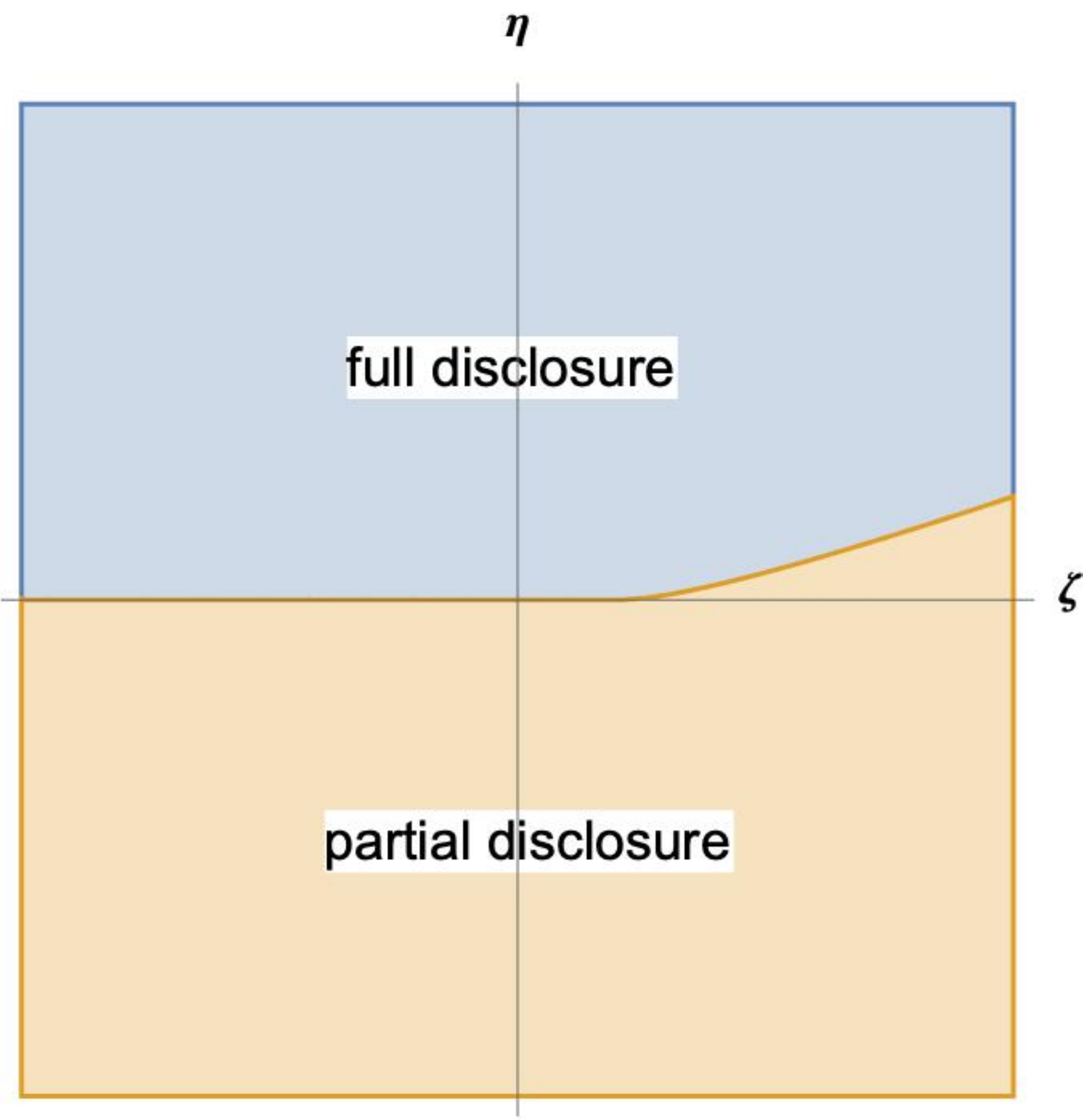}}\qquad \ 
  \subfloat[Exogenous private information]{\label{fig:exogenous}\includegraphics[width=3.55cm, bb=0 0 675 698]{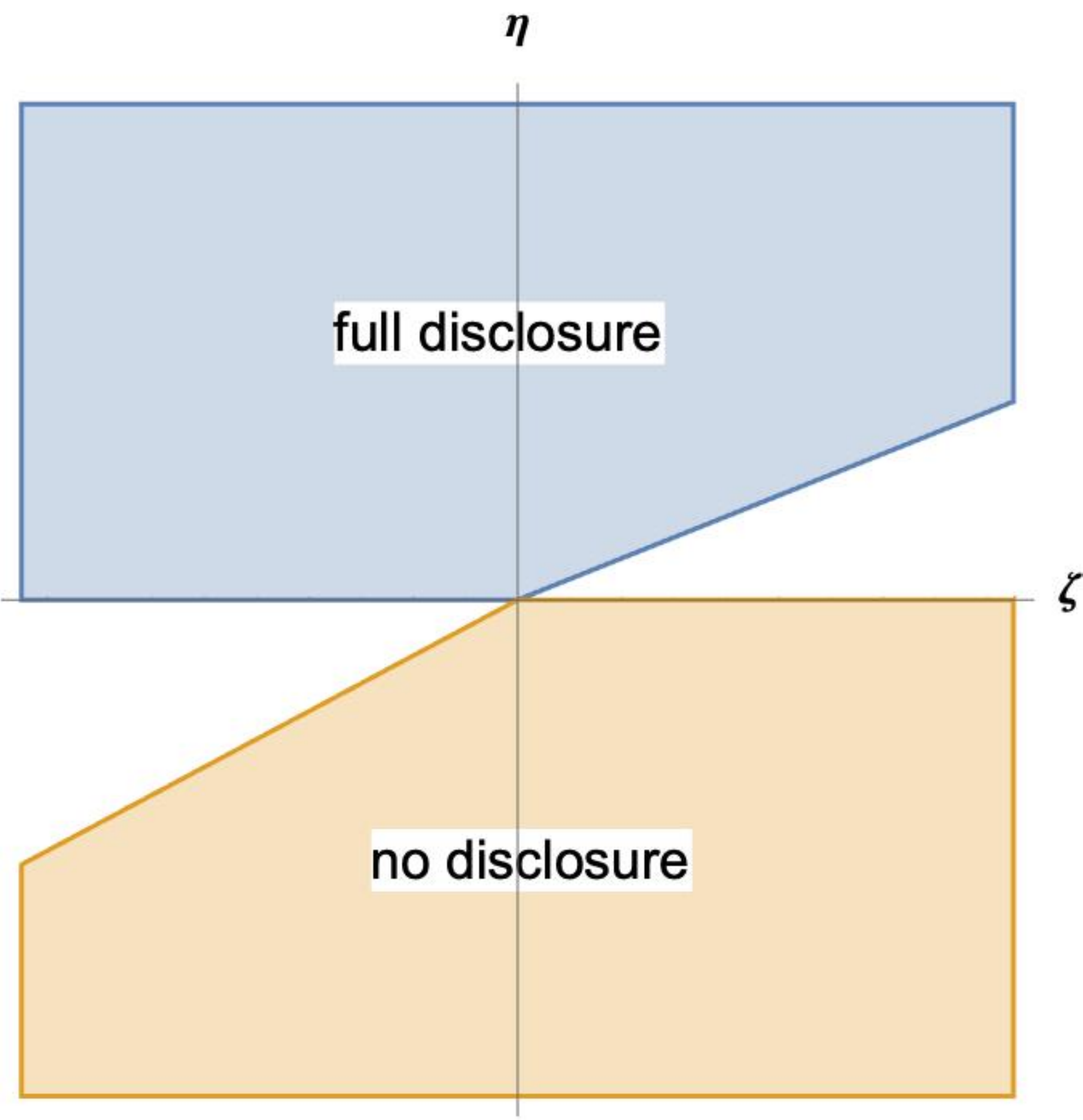}}
  \caption{The optimal disclosure on the $\zeta\eta$-plane. } 
  \label{fig2}
\end{figure}

Our last main result identifies the optimal precision of public information. 
We show that full disclosure (i.e., the precision is infinite) is optimal if the coefficient of volatility is strictly positive and greater than a threshold value determined by that of dispersion; partial disclosure (i.e., the precision is finite) is optimal otherwise (see Figure \ref{fig:flexible}). 
It should be noted that no disclosure is always suboptimal, which is in sharp contrast to the case of exogenous private information. \citet{uiyoshizawa2015} study optimal disclosure of public information under exogenous private information and show that no disclosure is optimal if the coefficient of volatility is small enough (see Figure \ref{fig:exogenous}). 
Even if the precision of private information is endogenously determined in the rigid information acquisition model  \citep{colombofemminis2014}, no disclosure can be optimal \citep{ui2022}.  

As applications of the above results, we consider a Cournot game \citep{vives1988} and an investment game \citep{angeletospavan2004} adopting the total net profit as a measure of firms' welfare. 
Mathematically, both games belong to the same class of games with different slopes of the best response: a Cournot game has a negative slope exhibiting strategic substitutability, and an investment game has a positive slope exhibiting strategic complementarity.
In the case of exogenous private information, 
more precise public information can reduce the total profit if and only if the game exhibits strong strategic {substitutability}, in which case no disclosure can be optimal \citep{morrisbergemann2013,angeletospavan2004}. 
This result remains valid in the case of rigid information acquisition \citep{ui2022}. 
A contrasting result holds in the case of flexible information acquisition. 
That is, more precise public information can reduce the total profit if and only if the game exhibits strong strategic {\em complementarity}, while full disclosure is always optimal. 
The difference arises from the crowding-out effect enhanced by strategic complementarity. 

We also apply our results to a beauty contest game \citep{morrisshin2002}. 
The welfare is the negative of the mean squared error of an individual action from the state minus the cost of information. 
In the case of exogenous private information, it is known that more precise public information can be harmful to welfare if the slope of the best response is greater than a threshold value \citep{morrisshin2002}. 
In the case of flexible information acquisition, a similar result holds, but the threshold value is less than in the case of exogenous private information; that is, welfare is more likely to decrease due to the crowding-out effect enhanced by strategic complementarity. 
This result is in clear contrast to the case of rigid information acquisition: the threshold value is greater than in the case of exogenous private information \citep{colombofemminis2008, ui2014, ui2022}.

This paper is organized as follows. 
We introduce the model and discuss an equilibrium of the second-stage subgame in Sections 2 and 3, respectively. 
The relationship between the total amount of information and the provision of public information is studied in Section~4. 
Section~5 derives the representation of welfare in terms of the volatility and dispersion. 
We study the welfare effect of public information in Section 6 and identify the optimal disclosure in Section~7. Section 8 is devoted to applications. All proofs are relegated to the appendices.

\subsection{Related literature}\label{Related literature}

Flexible information acquisition is modeled in the rational inattention framework \citep{sims2003}.\footnote{For a recent survey, see \citet{mackowiaketal2021}.} 
There are two types of flexible information acquisition in games. 
\citet{yang2015} considers independent information acquisition focusing on global games, where agents pay attention to a state alone. 
\citet{denti2020} considers correlated information acquisition, where agents pay attention not only to the state but also to the opponents' signals.  
\citet{rigos2020} studies LQG games under independent information acquisition,\footnote{The price-setting model of \citet{mackowiakwiederholt2009} is also an LQG game with independent information acquisition, but information acquisition technology is different.} and \citet{denti2020} (in an earlier version) studies them under correlated information acquisition. 
The second-period subgame of our model corresponds to the models of \citet{rigos2020} and \citet{denti2020}. 
We start with the assumption of correlated information acquisition and show that the set of equilibria in the second-period subgame is the same for both types of information acquisition. 

\citet{hebertlao2021} consider agents who flexibly acquire information about a state, a public signal, and the opponents' signals under a large class of information costs and study how properties of information costs relate to properties of equilibria.  
In particular, they examine whether agents correlate their actions with a public signal, which generates non-fundamental volatility, by assuming that a public signal is costly. 
In our paper, we assume that agents receive a public signal at no cost to study the direct impacts of public information provided by the policymaker. 

This paper contributes to the literature on the role of public information in symmetric LQG games with a continuum of agents.\footnote{LQG games are introduced by \citet{radner1962}. See \citet{vives2008} and \citet{ui2009,ui2016}.}  
In the case of exogenous private information, \citet{morrisshin2002} show that more provision of public information can be detrimental to welfare in a beauty contest game and attribute the harmful effect to agents' overreaction to a public signal under strategic complementarity.  
\citet{angeletospavan2007} study a symmetric LQG game and find a key factor for the social value of public information, which is referred to as the equilibrium degree of coordination relative to the socially optimal one. 
\citet{morrisbergemann2013} characterize the set of all Bayes correlated equilibria in this game and advocate a problem of finding optimal ones.\footnote{\citet{ui2020} considers information design in general LQG games by formulating it as semidefinite programming.}  
\citet{uiyoshizawa2015} solve the problem for an arbitrary quadratic welfare function by providing a necessary and sufficient condition for welfare to increase with private or public information. 
\citet{colombofemminis2014} integrate the symmetric LQG game and the model of rigid information acquisition with convex information costs and compare the equilibrium precision of private information with the socially optimal one.\footnote{Other models of information acquisition in symmetric LQG games are studied by \citet{hellwigveldkamp2009}, \citet{myattwallace2012,myattwallace2015}, among others.}   
\citet{ui2022} studies the welfare effect of public information in the model of \citet{colombofemminis2014} and identifies the optimal disclosure that maximizes welfare and the robust disclosure that maximizes the worst-case welfare, where the policymaker is uncertain about agents' information costs. 
In our paper, we assume flexible information acquisition and demonstrate that the impacts of public information are essentially different from the cases of exogenous private information and rigid information acquisition.  
We also show that the harmful effect of public information can be attributed to the crowding-out effect enhanced by strategic complementarity.

This paper is also related to the literature on Bayesian persuasion and information design\footnote{See survey papers by \citet{morrisbergemann2017} and \citet{kamenica2019} among others.} with a receiver who costly acquires information. 
\citet{lipnowskietal2020a} and \citet{bloedelsegal2020} consider a rationally inattentive receiver who can learn less information than what a sender provides.  In \citet{bloedelsegal2020}, a receiver pays costly attention to the sender's signals. In \citet{lipnowskietal2020a}, a receiver pays costly attention to a state, but the information is less than that provided by the sender. 
\citet{nizzottoetal2020} and \citet{matyskovamontes2021} consider a receiver who can acquire additional information after receiving a sender's signal at no cost, which is closer in spirit to our model. 
The models in the above papers are based on the Bayesian persuasion framework \citep{kamenicagentzkow2011}, and 
the sender's payoff is independent of the receiver's information cost.
Our model is an extension of a symmetric LQG game, and the policymaker's concern includes the agents' information costs.

\section{The model}\label{the model}

There are a policymaker and a continuum of agents indexed by $i\in [0,1]$. 
We consider the following two-period setting. 
In period~1, the policymaker chooses the precision of public information about an uncertain state and provides public information to agents at no cost. 
In period~2, each agent flexibly acquires private information at a cost proportional to the mutual information and chooses an action. 

Let $\theta$ denote an uncertain state, which is normally distributed with mean $\bar\theta$ and precision $\tau_\theta$ (the reciprocal of the variance).
Public information is given as a random variable $\tilde\theta$, which is normally distributed such that 
the posterior distribution of $\theta$ given $\tilde \theta$ is a normal distribution with mean $\tilde\theta$ and precision $\tau\geq\tau_\theta$; that is, $\E[\theta|\tilde\theta]=\tilde\theta$ and $\var[\theta|\tilde\theta]=1/\tau\leq 1/\tau_{\theta}$. Such a random variable $\tilde\theta$ exists and statisfies $\var[\tilde\theta]=\cov[\tilde\theta,\theta]=1/\tau_\theta-1/\tau$. 
In particular, if $\tau=\tau_\theta$, then $\var[\tilde\theta]=0$. 
We call $\tau$ the precision of public information, which is chosen by the policymaker in period 1.\footnote{In the literature, public information is often given as $y=\theta+\varepsilon$, where $\varepsilon$ is a normally distributed noise term with precision $\tau_y$. The conditional distribution of $\theta$ given $y$ is a normal distribution with mean $\E[\theta|y]$ and precision $\tau_\theta+\tau_y$. Thus, we can translate $y$ into $\tilde\theta$ by letting $\tilde \theta=\E[\theta|y]$ and $\tau=\tau_\theta+\tau_y$ because $\E[\theta|y]=\E[\theta|\E[\theta|y]]$.}

Agent $i$'s action is a real number $a_i \in \mathbb{R}$. 
We write $a=(a_i)_{i\in [0,1]}$, $a_{-i}=(a_j)_{j\neq i}$, and $A=\int a_idi$. 
Agent $i$'s payoff to an action profile $a$ is
\begin{align}
u_i({a}, \theta)=-a_i^2+2\alpha a_i A+2\beta \theta a_i+h(a_{-i},\theta), \label{payoff function cont}
\end{align}
which is symmetric and quadratic in $a$ and $\theta$. 
This game exhibits strategic complementarity if $\alpha>0$ and strategic substitutability if $\alpha<0$. We assume $\alpha<1$, which guerantees the existence and uniquenss of a symmetric equilibrium under a symmetric Gaussian information structure. We also assume $\beta >0$ without loss of generality.
By completing the square, we obtain  
\begin{align}
u_i({a}, \theta)=
-\left(a_i-(\alpha A+\beta\theta)\right)^2+h'(a_{-i},\theta), \label{payoff function cont 1}
\end{align}
where $h'(a_{-i},\theta)=h(a_{-i},\theta)-(\alpha A+\beta\theta)^2$. 
Thus, agent $i$'s best response equals the conditional expectation of a target $\alpha A+\beta\theta$. We call $\alpha$ the slope of the best response.  
Note that $h'(a_{-i},\theta)$ has no influence  on the best response, so we can assume $h'(a_{-i},\theta)=0$ when we study equilibria.


Agent $i$ acquires private information about a target $\alpha A+\beta\theta$ \citep{denti2020}. 
The cost of information is proportional to the mutual information of a private signal $s_i$ and a target $\alpha A+\beta\theta$. 
The mutual information of two random variables $x$ and $y$ is denoted by
\[
\I(x;y)\equiv \int p(x,y)\log \frac{p(x,y)}{p(x)p(y)}dxdy,
\] 
where $p(x,y)$ is the joint density function of $(x,y)$, and 
$p(x)$ and $p(y)$ are the marginal density functions of $x$ and $y$, respectively. 
We write $\It(x;y)$ for the mutual information with respect to the conditional distribution of $(x,y)$ given a public signal $\tilde\theta$: 
\[
\It(x;y)\equiv \int p(x,y|\tilde\theta)\log \frac{p(x,y|\tilde\theta)}{p(x|\tilde\theta)p(y|\tilde\theta)}dxdy. 
\] 
Then, the cost of a private signal $s_i$ is 
\begin{equation}
C(s_i)\equiv 
\lambda \cdot \It(\alpha A+\beta\theta;s_i),
\notag\label{def of cost}	
\end{equation}
where $\lambda>0$ is constant. 

Without loss of generality, we restrict attention to a direct signal $s_i=a_i$ that suggests an optimal action, as justified by the revelation principle. 
Agent $i$'s direct signal is optimal if it is a solution of an optimization problem 
\begin{equation}
\max_{a_i}-\Et[\left(a_i-(\alpha A+\beta\theta)\right)^2]-C(a_i) 
\label{optimal signal}
\end{equation}
where $a_i$ follows an arbitrary distribution, and $\Et$ denotes the conditional expectation operator when a public signal $\tilde\theta$ is given.  
An optimal direct signal  must satisfy $a_i=\Et[\alpha A+\beta \theta|a_i]$.  


Anticipating the agents' decision in period 2, the policymaker chooses $\tau$ in period 1 to maximize the expected value of a welfare function given by  
\begin{equation}
v(a,\theta)-\int C(a_i)di, \label{bp's payoff1}
\end{equation}
where $v(a,\theta)$ is symmetric and quadratic in $a$ and $\theta$, i.e.,  
\begin{equation}
v(a,\theta)\equiv c_1\int_0^1 a_j^2 dj+c_2\left(\int_0^1 a_j dj\right)^2+c_3\theta\int_0^1 a_j dj+c_4\int_0^1 a_j dj+c_5. \label{v function}
\end{equation}
A typical case is the aggregate payoff, i.e., $v(a,\theta)=\int u_i(a,\theta) di$.

\section{The second-period subgame}\label{section: The second-period subgame}

This section focuses on the second-period subgame and calculates an equilibrium. 
\citet{denti2020} (in an earlier version) obtains an equilibrium 
assuming $\tau< 2\beta^2/\lambda$, which guarantees the existence and uniquness of equilibrium. 
We drop this assumption and obtain the set of equilibria for all $\tau$ using a parameter called the information fraction \citep{bowsherswain2012}, which will play an essential role in our analysis.


In the second-stage subgame, a public signal $\tilde\theta$ with precision $\tau$ is given, and it is common knowledge that $\theta$ is normally distributed with mean $\tilde\theta$ and precision $\tau$. 
Thus, throughout this section, every probability distribution is understood as the conditional one given $\tilde\theta$. 
Let $\vart[x]$ and $\covt[x,y]$ denote the conditional variance of $x$ and the conditional covariance of $x$ and $y$, respectively. For example, $\vart[\theta]=1/\tau$. 
Henceforth, {\em conditional} will be omitted until the end of this section.

We focus on a symmetric equilibrium in which $(a_i,A,\theta)$ is jointly normally distributed. 
A joint distribution of $(a_i,A,\theta)$ is said to be an equilibrium if it satisfies the following conditions. 
\begin{enumerate}
	\item $(a_i,A,\theta)$ is jointly normaly distributed, where $\Et[\theta]=\tilde\theta$ and $\vart[\theta]=1/\tau$. 
	\item $a_i$ is an optimal direct signal. 
	\item $\Et[a_i]=\Et[A]$, $\vart[A]=\covt[a_i,A]$, and $\covt[A,\theta]=\covt[a_i,\theta]$.
\end{enumerate}
The last condition must hold because $A=\int a_idi$ and an equilibrium is symmetric, i.e., for all $i\neq j$, 
$\Et[a_i]=\Et[a_j]$, $\covt[a_i,A]=\covt[a_j,A]$, and $\covt[a_i,\theta]=\covt[a_j,\theta]$. 


We characterize an equilibrium using the ratio of the variance of an action $a_i$ to that of a target $\alpha A+\beta\theta$, which is denoted by  
\begin{equation}
\gamma\equiv \vart[a_i]/\vart[\alpha A+\beta\theta]=\vart[\Et[\alpha A+\beta\theta|a_i]]/\vart[\alpha A+\beta\theta]\leq 1
\notag \label{def of t}	
\end{equation}
and referred to as the information fraction \citep{bowsherswain2012}. 
It is known that the information fraction equals the square root of the correlation coefficient between $a_i$ and $\alpha A+\beta \theta$. 
Thus, it measures the amount of information about $\alpha A+\beta\theta$ contained in $a_i$. In fact, the mutual information of $a_i$ and $\alpha A+\beta\theta$ is represented as 
\begin{equation}
-\frac{\log(1-\gamma)}{2}\notag \label{cost formula}	
\end{equation}
when $a_i$ and $\alpha A+\beta\theta$ are jointly normally distributed. 
Note that $\gamma>0$ if and only if $\vart[a_i]>0$ and that $\gamma<1$ because the cost of information is infinite when $\gamma=1$.  

The next proposition shows that $\gamma$ satisfies either 
\begin{equation}
\tau=f(\gamma)\equiv 
\frac{2 \beta ^2}{\lambda}\cdot\frac{1-\gamma}{(1-\alpha  \gamma)^2}\label{main result eq 2}
\end{equation}
or $\gamma=0$ with $\tau\geq f(0)$ in every equilibrium. 

\begin{proposition}\label{proposition 1}
There exists an equilibrium with $\gamma>0$ if and only if there exists $\gamma>0$ satisfying \eqref{main result eq 2}. 
There exists an equilibrium with $\gamma=0$ if and only if $\tau\geq f(0)$. 
The number of equilibria is determined by $\alpha$, $\beta$, $\lambda$, and $\tau$.
\begin{enumerate}[\em(i)]

\item Suppose that $\alpha\leq 1/2$. Then, an equilibrium is unique: $\gamma>0$ if $\tau< f(0)$ and $\gamma=0$ if $\tau\geq f(0)$. 
\item Suppose that $\alpha> 1/2$. 
\begin{enumerate}[\em (a)]
\item If $\tau< f(0)$ or $\tau>f(({2 \alpha -1})/{\alpha })$, 
an equilibrium is unique: $\gamma>0$ in the former case and $\gamma=0$ in the latter case. 
\item If $\tau= f(0)$ or $\tau=f(({2 \alpha -1})/{\alpha })$, two equilibria exist: $\gamma=0$ in one of them. 
\item If $f(0)<\tau<f(({2 \alpha -1})/{\alpha })$, three equilibria exist: $\gamma=0$ in one of them. 
\end{enumerate}
\end{enumerate}
In every equilibrium, $\Et[a_i]=\Et[A]=\beta \tilde\theta/(1-\alpha)$. 
In an equilibrium with $\tau=f(\gamma)$, 
\begin{gather}
	A-\Et[A]=\beta \gamma(\theta-\tilde\theta)/(1-\alpha \gamma),\label{main equi 0}\\
	\vart[a_i]=\frac{\lambda \gamma}{2(1-\gamma)},
	\label{main equi 1}\\
	 \covt[a_i,A]=\vart[A]=\gamma\vart[a_i]=\frac{\lambda \gamma^2}{2(1-\gamma)}, \label{main equi 2}\\
	\covt[a_i,\theta]=\covt[A,\theta]=(\vart[a_i]-\alpha\covt[a_i,A])/\beta=\frac{\lambda  \gamma (1-\alpha  \gamma)}{2 \beta  (1-\gamma)}. \label{main equi 3}
\end{gather}
\end{proposition}

 This proposition says that there are two types of equilibria depending upon $\tau$ and $\gamma$. 
In the first type with $\tau=f(\gamma)$, agents acquire information. 
In the second type with $\gamma=0$ and $\tau\geq f(0)$, agents do not acquire information. 
For later use, we calculate the set of possible information fractions in the first type (a straightforward proof is omitted), which is illustrated in Figure \ref{fig3}. 
\begin{lemma}\label{def of gamma}
For each $\tau$, the set of information fractions satisfying $\tau=f(\gamma)$ is
\begin{equation}
\Gamma(\tau)\equiv \{\gamma\in [0,1]\mid \tau=f(\gamma)\}=
\begin{cases}
\{\overline\phi(\tau)\} & \text{ if } \tau< f(0) \text{ or if }\alpha\leq 1/2 \text{ and }\tau= f(0),\\
\{\overline\phi(\tau),\underline\phi(\tau)\} & \text{ if } \alpha>1/2 \text{ and }f(0)\leq\tau\leq f((2\alpha-1)/\alpha),\\
\emptyset &\text{ otherwise,}
\end{cases}
\notag
\end{equation}
where 
\[
\overline\phi(\tau)\equiv\frac{\alpha  \lambda  \tau -\beta ^2+\beta  \sqrt{\beta ^2-2 (1-\alpha ) \alpha  \lambda  \tau }}{\alpha ^2 \lambda  \tau },\ 
\underline\phi(\tau)\equiv\frac{\alpha  \lambda  \tau -\beta ^2-\beta  \sqrt{\beta ^2-2 (1-\alpha ) \alpha  \lambda  \tau }}{\alpha ^2 \lambda  \tau }.
\]
In addition, $\underline\phi(\tau)\leq(2\alpha-1)/\alpha\leq \overline\phi(\tau)$ if  $\alpha>1/2$ and 	$f(0)\leq\tau\leq f((2\alpha-1)/\alpha)$.
\end{lemma}

\begin{figure} 
  \centering
  \subfloat[$\alpha\leq 1/2$. $f(\gamma)$ is decreasing.]{\label{fig:case (i)}
  \includegraphics[width=4cm, bb=0 0 675 440]{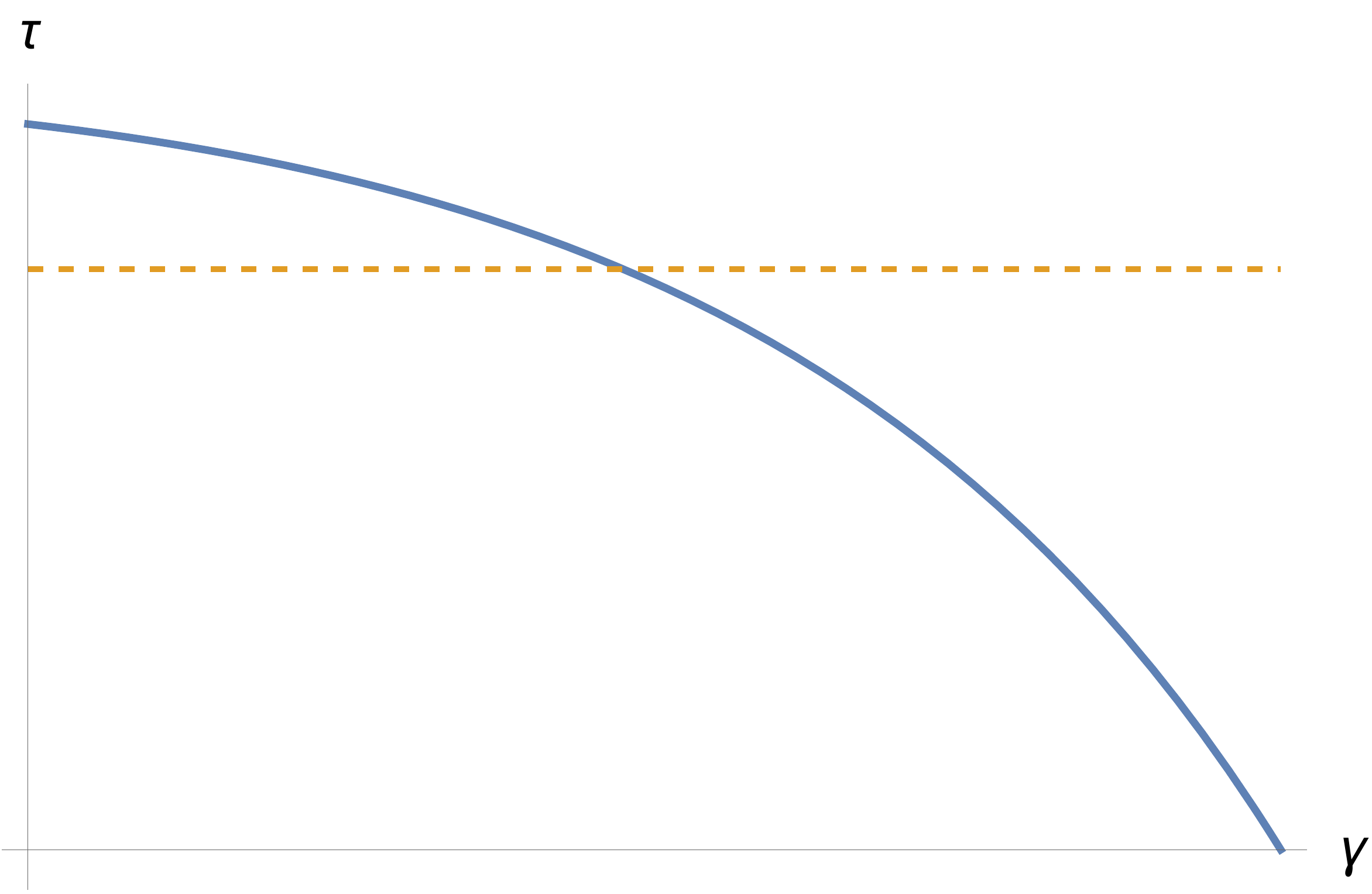}}\qquad \ 
  \subfloat[$\alpha > 1/2$. $f(\gamma)$ has a maximum at $\gamma=(2\alpha-1)/\alpha$.]{\label{fig:case (ii)}\includegraphics[width=4cm, bb=0 0 675 440]{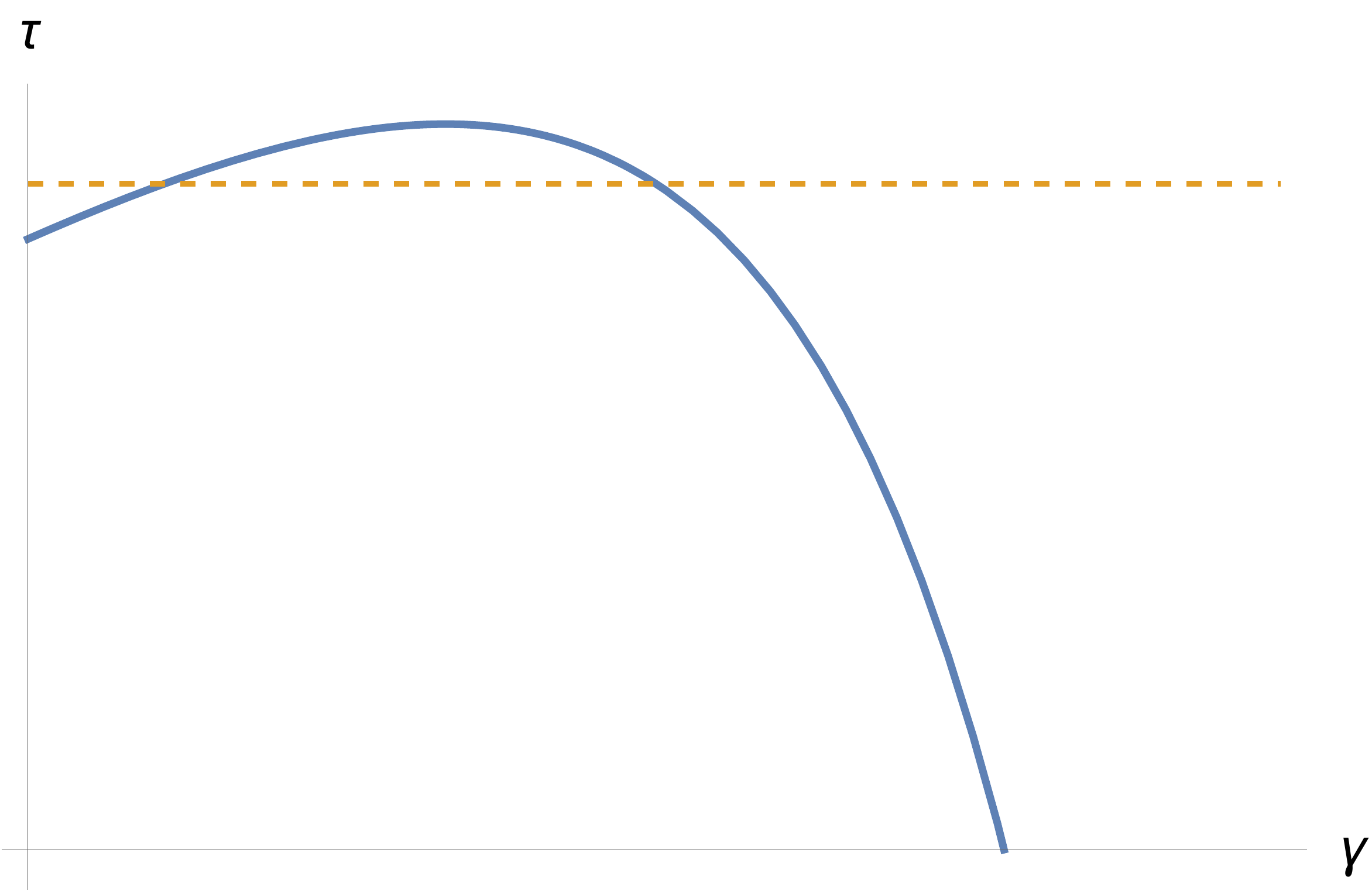}}
  \caption{The graphs of $\tau=f(\gamma)$ on the $\gamma\tau$-plane.} 
  \label{fig3}
\end{figure}


 The number of equilibria also depends upon the slope of the best response $\alpha$. 
If the slope is not too large (i.e., $\alpha\leq 1/2$), 
an equilibrium is unique. 
However, if the slope is large enough (i.e., $\alpha> 1/2$), multiple equilibria arise,   
which results from strong strategic complementarity. 
In a game with strategic complementarity, not only an action but also information is a strategic complement  \citep{hellwigveldkamp2009}. 
Thus, it can be an equilibrium strategy to acquire a small amount of information as well as a large amount of information when $f(0)<\tau<f(({2 \alpha -1})/{\alpha })$. 
On the other hand, an equilibrium is unique in the following cases. 
When $\tau<f(0)$, each agent has a strong incentive to acquire a large amount of information under considerable uncertainty, irrespective of how much information the opponents possess; when $\tau<f(0)$, each agent has no incentive to acquire information under little uncertainty.

In an equilibrium with $\gamma>0$, $A$ and $\alpha A+\beta\theta$ are completely determined by and affine in $\theta$, as \eqref{main equi 0} implies.  
Thus, actions are conditionally independent given $\theta$ because, otherwise, $A$ must also depend on the correlated component of actions. 
In addition, the mutual information of $a_i$ and $\alpha A+\beta\theta$ equals that of $a_i$ and $\theta$. 
This means that Proposition \ref{proposition 1} remains valid under the alternative assumption of independent information acqusition \citep{yang2015} that 
\begin{itemize}
	\item agents acquire private information about a state $\theta$ (rather than a target $\alpha A+\beta\theta$) at a cost proportional to the mutual information, i.e., $C(s_i)\equiv \lambda \cdot \It(\theta;s_i)$,
\item private signals are conditionally independent given $\theta$. 
\end{itemize}
\citet{rigos2020} calculates the set of equilibria in this case, which is shown to coincide with that in Proposition \ref{proposition 1}. 

Because \citet{denti2020} and \citet{rigos2020} do not show the equivalence of the two equilibrium concepts and do not use the information fraction to represent an equilbrium, we provide a proof for Proposition \ref{proposition 1}. 
It is an immediate consequence of the following well-known result.\footnote{This result is a restatement of Theorem 10.3.2 of \citet{coverthomas2006} and its proof.}

\begin{lemma}\label{rate distortion lemma}
Let $x$ be a normally distributed random variable with mean $\bar x$ and variance $\sigma^2$. 
Consider an optimization problem  
\begin{equation}
\max_y-\E[(x-y)^2]-\lambda\cdot \I[x;y], \label{CT lemma}
\end{equation} 
where $y$ is a random variable with an arbitrary distribution.
An optimal solution exists and satisfies the following.
\begin{itemize}
\item If $\lambda/2<\sigma^2$, then $y$ is a normally distributed random variable satisfying $\E[x|y]=y$, $\E[(x-y)^2]=\var[x]-\var[y]=\lambda/2$, and $\I(x;y)=1/2 \cdot\log (2\sigma^2/\lambda)$. 
\item If $\lambda/2\geq\sigma^2$, then $y=\bar x$ and $\I(x;y)=0$.
\end{itemize}
\end{lemma}

Using this lemma, we can identify an optimal direct signal because \eqref{optimal signal} is \eqref{CT lemma} with $x=\alpha A+\beta\theta$ and $y=a_i$ when $A$ and $\theta$ are jointly normally distributed. 
For example, consider a special case of no strategic interaction, i.e., $\alpha=0$, 
where a target is $\beta\theta$. 
By Lemma \ref{rate distortion lemma}, 
if $\lambda/2<\vart[\beta\theta]=\beta^2/\tau$, then 
$\vart[a_i]=\vart[\beta\theta]-\lambda/2=\vart[a_i]/\gamma-\lambda/2$. 
By solving this, we obtain $\vart[a_i]={\lambda \gamma}/({2(1-\gamma)})$ and 
$\vart[\beta\theta]=\beta^2/\tau=\vart[a_i]/\gamma={\lambda}/({2(1-\gamma)})$. 
This is rewritten as 
$\tau={2 \beta ^2}({1-\gamma})/\lambda$, 
which is a special case of \eqref{main result eq 2} with $\alpha=0$. 
\begin{remark}
We assume the mutual information cost function in Proposition \ref{proposition 1}. 
The proposition remains valid under the assumption of other cost functions such as the Fisher information cost function \citep{hebertwoodford2021}. 
See Appendix \ref{Flexible information acquisition with the Fisher information costs}.
\end{remark}

\section{Total information}\label{The crowding-out effect of public information}

Each agent receives two types of information about $\theta$, public information $\tilde\theta$ in period 1 and private information $a_i$ in period 2. 
When the policymaker provides more precise public information, agents acquire less precise private information in a unique equilibrium with information acquisition. 
This effect of public information on private information is referred to as the crowding-out effect. 
In the presence of the crowding-out effect, does the total amount of information about $\theta$ increase with public information? 
We address this question in terms of mutual information. 

Assume that $\tau_\theta$ is sufficiently small, i.e.,  $\tau_\theta<f(0)$. 
Imagine that the policymaker chooses $\tau<f(0)$ in period 1 and agents follow the unique equilibrium in period 2, where the information fraction is $\gamma=\overline\phi(\tau)$ by Lemma \ref{def of gamma}. 
We measure the total amount of information about $\theta$ contained in $(\tilde\theta,a_i)$ by means of $\I(\tilde\theta,a_i; \theta)$, i.e., the mutual information of $(\tilde\theta,a_i)$ and $\theta$. 
By the chain rule of information, it holds that 
\begin{equation}
\I(\tilde\theta,a_i; \theta)=\I(\tilde\theta;\theta)+\E[\It(a_i;\theta)], \notag \label{total amount of information}
\end{equation}
where $\I(\tilde\theta;\theta)$ is the mutual information of $\tilde\theta$ and $\theta$, and 
$\It(a_i;\theta)$ is the mutual information of $a_i$ and $\theta$ under the conditional probability distribution given $\tilde\theta$. 
That is, the total amount of information is the sum of those of public information and private information, each of which is calculated as 
\begin{equation}
\I(\tilde\theta;\theta)=-\frac{\log \tau_\theta/\tau}{2},\ 
\E[\It(a_i;\theta)]=\It(a_i;\theta)=-\frac{\log (1-\overline\phi(\tau))}{2}.
\label{amount of each information}	
\end{equation}
Clearly, the public-information component $\I(\tilde\theta;\theta)$ is increasing in $\tau$. 
In contrast, the private-information component $\It(a_i;\theta)$ is decreasing in $\tau$ because the information fraction $\overline\phi(\tau)$ is deceasing in $\tau$. 

\begin{lemma}\label{lemma: crowding-out}
If $\tau<f(\max\{0,(2\alpha-1)/\alpha\})$, then $\overline\phi'(\tau)<0$.
\end{lemma}

By \eqref{amount of each information}, the total amount of information is 
\begin{equation}
\I(\tilde\theta,a_i;\theta)={I}_{\overline\phi}(\tau)\equiv
\frac{1}{2}\left(\log\frac{\tau}{1-\overline\phi(\tau)}-\log\tau_\theta\right)
=\frac{1}{2}\left(\log\frac{2 \beta ^2}{\lambda(1-\alpha  \overline\phi(\tau))^2}-\log\tau_\theta\right),
\label{total amount of information 2}
\end{equation}
where the second equality follows from $f(\overline\phi(\tau))=\tau$ and \eqref{main result eq 2}. 
By differentiating \eqref{total amount of information 2} with respect to $\tau$, we find that 
the total amount of information increases with the precision of public information if the game exhibits strategic substitutability and decreases if it exhibits strategic complementarity, as shown in the next proposition. 

\begin{proposition}\label{main proposition 0}
If $\tau<f(\max\{0,(2\alpha-1)/\alpha\})$ and $\gamma=\overline\phi(\tau)$, then 
\begin{equation}
\frac{d\I(\tilde\theta,a_i;\theta)}{d\tau}=\frac{d{{I}_{\overline\phi}}(\tau)}{d\tau}=\frac{\alpha\overline\phi'(\tau)}{1-\alpha\overline\phi(\tau)}. \label{main prop 1 eq}
\end{equation}
Thus, ${I}_{\overline\phi}$ 
is increasing in $\tau$ if $\alpha<0$, decreasing in $\tau$ if $\alpha>0$, and independent of $\tau$ if $\alpha=0$. 
\end{proposition}

To provide another implication of \eqref{main prop 1 eq}, 
we divide both sides by ${d\I(\tilde\theta;\theta)}/{d\tau}$ and obtain  
\begin{equation}
\frac{d\It(a_i;\theta)}{d\tau}\ \Big/ \ \frac{d\I(\tilde\theta;\theta)}{d\tau}+1=
\frac{2\alpha\tau\overline\phi'(\tau)}{1-\alpha\overline\phi(\tau)}. \notag
\end{equation}
This implies that 
\begin{equation}
\mu\equiv -
\frac{d\It(a_i;\theta)}{d\tau}\, \Big/ \  \frac{d\I(\tilde\theta;\theta)}{d\tau}\gtreqless 1\ \Leftrightarrow\ \alpha\gtreqless 0.\label{MRS property}	
\end{equation}
We refer to $\mu$ as the marginal rate of substitution of public information for private information, which measures the crowding-out effect. 
The marginal rate of substitution is the change in the amount of private information $\It(a_i;\theta)$ resulting from a one-unit increase in the amount of public information $\I(\tilde\theta;\theta)$. 
If $\mu$ is very large, a one-unit increase in public information substantially decreases private information, i.e., the crowding-out effect is significant. 
By \eqref{MRS property}, the crowding-out effect is larger in a game with strategic complementarity than strategic substitutability, which is another implication of Proposition \ref{main proposition 0}.\footnote{\citet{colombofemminis2014} obtain a similar result in the model of rigid information acquisition. See Remark \ref{remark 1}.} 

In the benchmark case of $\alpha=0$, the marginal rate of substitution equals one; that is, public and private signals are perfect substitutes when there is no strategic interaction. 
Thus, when the policymaker increases the amount of public information, 
an agent reduces the amount of private information so that the total amount of information can remain constant.

If $\alpha\neq 0$, however, the marginal rate of substitution is not constant; that is, public and private signals are imperfect substitutes. 
If the game exhibits strategic complementarity (i.e., $\alpha>0$), the crowding-out effect is stronger, and the total amount of information decreases with public information. 
This is because agents put more value on public information to align their actions with the aggregate action,\footnote{Indeed, the weight of $\tilde\theta$ in $a_i$ is $\Et[a_i]/\tilde\theta=\beta/(1-\alpha)$ by Proposition \ref{proposition 1}, which is increasing in $\alpha$.} so that a one unit increase in more valuable public information compensates for a larger decrease in private information. 
In contrast, if the game exhibits strategic substitutability (i.e., $\alpha<0$), the crowding-out effect is weaker, and the total amount of information increases with public information. 
This is because agents put less value on public information so that a one-unit increase in less valuable public information compensates for a smaller decrease in private information. 

Before closing this section, we ask whether the marginal rate of substitution is monotone in $\alpha$, as suggested by \eqref{MRS property}.  
The answer depends upon which parameter to fix when $\alpha$ changes.  
By \eqref{main prop 1 eq}, $\mu$ is represented as a function of $(\alpha,\tau)$: 
\begin{equation}
\mu=\mu_1(\alpha,\tau)\equiv 1-\frac{2\alpha\tau\overline\phi'(\tau)}{1-\alpha\overline\phi(\tau)}. \notag
\end{equation}
We can verify that $\mu_1(\alpha,\tau)$ is not monotone in $\alpha$; that is, the marginal rate of substitution is not necessarily increasing in $\alpha$ for fixed $\tau$. 
However, by plugging $\tau=f(\gamma)$ into $\mu_1(\alpha,\tau)$, 
we can represent $\mu$ as a function of $(\alpha,\gamma)$, which is monotone in $\alpha$: 
\begin{equation}
\mu=\mu_2(\alpha,\gamma)\equiv \frac{1-\alpha  \gamma }{1-\alpha  (2-\gamma)},\ 
\frac{\partial\mu_2(\alpha,\gamma)}{\partial\alpha}=\frac{2 (1-\gamma)}{(1-(2-\gamma)\alpha)^2}>0.
\label{monotone substitutability}	
\end{equation}
That is, the marginal rate of substitution is increasing in $\alpha$ for fixed $\gamma$. 
Consequently, when evaluated at the same information fraction, the crowing-out effect is larger in a game with a stronger degree of strategic complementarity. 
This observation will be used to provide intuition for results in the next section.

\begin{remark}
Proposition \ref{main proposition 0} (together with the subsequent discussion) remains valid under the assumption of other cost functions such as the Fisher information cost function \citep{hebertwoodford2021}. 
See Appendix \ref{Flexible information acquisition with the Fisher information costs}.
\end{remark}

\begin{remark}\label{remark 1}
\citet{colombofemminis2014} study the model of rigid information acquisition and show that the substitutability between public and private information is increasing in the slope of the best response. 
In Appendix \ref{Rigid information acquisition with linear information costs},  
we discuss the case of rigid information acquisition with linear information costs and show that the total amount of information decreases with public information if and only if the game exhibits strategic complementarity, which is the same as Proposition~\ref{main proposition 0}. 
We also show that the crowding-out effect under flexible information acquisition is more significant than that under rigid information acquisition if and only if the game exhibits strategic complementarity. 
In other words,  the crowding-out effect is more sensitive to the slope of the best response in the model of flexible information acquisition. 
\end{remark}

\begin{remark}
We can conduct a similar analysis in the case of multiple equilibria, where $\alpha>1/2$ and $f(0)<\tau<f((2\alpha-1)/\alpha)$. 
Note that $\Gamma(\tau)=\{\underline\phi(\tau),\overline\phi(\tau)\}$ by Lemma \ref{def of gamma}. 
When $\gamma=\overline\phi(\tau)$, the above results  
remain valid; that is, the total amount of information decreases with public information because $\alpha>0$. 
When $\gamma=\underline\phi(\tau)$, however, we obtain the opposite result; that is, the total amount of information increases with public information. 
This is because agents acquire more precise private information when the policymaker provides more precise public information in the case of $\gamma=\underline\phi(\tau)$. 
See Appendix \ref{total information under multiple equilibria}.  
\end{remark}

\section{The definition of optimality}\label{The definition of optimality}

In this section, we calculate the policymaker's objective function and provide our definition of optimal disclosure. 
When the policymaker chooses~$\tau$ and the agents follow equilibrium strategies, 
the expected welfare is calculated as 
\begin{align}
\E[v(a,\theta)]- C(a_i)&=
c_1\var[a_i]+c_2\var[A]+c_3\cov[\theta,A]- C(a_i)+d\label{welfare 0}
\end{align}
by \eqref{v function}, where $d$ is constant, because the expected values of $a_i$ and $A$ are constants independent of $\tau$ and $\gamma$ by Proposition \ref{proposition 1}. 
We rewrite \eqref{welfare 0} using volatility and dispersion of actions, which follows \citet{uiyoshizawa2015} who characterize optimal disclosure in the case of exogenous private information. 
The volatility is the variance of the average action, $\var[A]$, 
and the dispersion is the variance of the idiosyncratic difference, $\var[a_i-A]=\var[{a_i}]-\var[A]$ \citep{angeletospavan2007,morrisbergemann2013}.

Note that the first three terms in \eqref{welfare 0} are linearly dependent because 
\begin{equation}
\cov[A,\theta]=(\var[a_i]-\alpha\cov[a_i,A])/\beta=(\var[a_i]-\alpha\var[A])/\beta
\notag \label{covthetaa}
\end{equation}
by \eqref{main equi 3} and the law of total variance. 
Thus, we can write \eqref{welfare 0} by means of volatility and dispersion, 
\begin{equation}
\E[v(a,\theta)]- C(a_i)=
\zeta \var[a_i-A]+\eta\var[A]- C(a_i)+d,\label{welfare 1}
\end{equation}
where $\zeta=c_1+c_3/\beta$ and $\eta=c_1+c_2+(1-\alpha)c_3/\beta$. 
Consequently, the expected welfare is calculated as follows.

\begin{lemma}\label{welfare lemma}
When the policymaker chooses $\tau$ and the agents follow an equilibrium with $\gamma$,  
the dispersion and volatility are given by 
\begin{align}
D(\tau,\gamma)&=
\begin{cases}
D_+(\gamma)\equiv {\lambda\gamma}/{2} & \text{ if $\tau=f(\gamma)$},\\	
0  & \text{ if $\tau\geq f(0)$ and $\gamma=0$},
\end{cases}\notag\\
V(\tau,\gamma)&=
\begin{cases}
V_+(\gamma)\equiv  (\beta^2\tau_\theta^{-1}-\lambda/{2}\cdot ((1-2 \alpha ) \gamma+1))/(1-\alpha)^2 & \text{ if $\tau=f(\gamma)$},\\	
\displaystyle V_0(\tau)\equiv  { \beta^2(\tau_\theta^{-1}-\tau^{-1})}/{(1-\alpha)^2} & \text{ if $\tau\geq f(0)$ and $\gamma=0$},
\end{cases}\notag
\end{align}
respectively, and the cost of information is given by 
$
C(a_i)=-\lambda/2\cdot \log(1-\gamma)$. Thus, the expected welfare equals $W(\tau,\gamma)+d$, where 
\begin{align}
W(\tau,\gamma)&=
\begin{cases}
 \displaystyle
W_+(\gamma) \equiv \zeta D_+(\gamma)+\eta V_+(\gamma)-\lambda/2\cdot \log(1-\gamma)^{-1}
& \text{ if $\tau=f(\gamma)$},
 \\	
\displaystyle W_0(\tau) \equiv
 \eta V_0(\tau) & \text{ if $\tau\geq f(0)$ and $\gamma=0$}.
\end{cases}
\notag\label{welfare simple}
\end{align}
\end{lemma}

Henceforth, we regard $W(\tau,\gamma)$ as the expected welfare, which equals either that under information acquisition, $W_+(\gamma)$, or that under no information acquisition, $W_0(\tau)$. 
Note that $W(\tau,\gamma)$ is a function of $\tau$ and $\gamma$, and  
 $\gamma$ is uniquely determined by $\tau$ if $\alpha\leq 1/2$ or $\tau<f(0)$. 
 If $\alpha> 1/2$ and $f(0)\leq\tau<f(({2 \alpha -1})/{\alpha })$, however, multiple equilibria arise and $\gamma$ is not unique. 
In this case, we assume that the agents follow an equilibrium with the largest welfare (i.e., a sender-optimal equilibrium), following the standard Bayesian persuasion framework \citep{kamenicagentzkow2011}.

The optimal precision maximizes the expected welfare under either information acquisition or no information acquisition. 
In the case of information acquisition (i.e., $\tau=f(\gamma)$), the optimal precision is given by
\[
T_+^*\equiv\arg\max_{\tau\geq\tau_\theta:\ \Gamma(\tau)\neq\emptyset} \overline{W}_+(\tau),
\]	
where $\overline{W}_+(\tau)\equiv\max_{\gamma\in\Gamma(\tau)}W_+(\gamma)$ is the maximum achievable welfare given $\tau$ under information acquisition. 
The domain of $\overline{W}_+(\tau)$ is 
$\{\tau\geq\tau_\theta\mid \Gamma(\tau)\neq\emptyset\}=[\tau_\theta,\bar\tau]$,
where 
\[
\bar\tau\equiv \max_{\gamma\in [0,1)}f(\gamma)=
\begin{cases}
f(0) & \text{ if }\alpha\leq 1/2,\\
f((2\alpha-1)/\alpha) & \text{ if }\alpha> 1/2.
\end{cases}
\]
We refer to $\bar\tau$ as the maximum precision under information acquisition. 
In the case of no information acquisition (i.e., $\gamma=0$ and $\tau\geq f(0)$), the optimal precision is given by 
\[
T_0^*=\arg\max_{\tau\in [f(0),\infty]}W_0(\tau).
\]
The globally optimal precision is defined as follows.
\begin{definition}\label{def optimal}
We say that $\tau^*$ is optimal if one of the following holds.
\begin{itemize}
\item $\tau^*\in T_+^*$ and 
$\overline{W}_+(\tau^*)\geq W_0(\tau)$ for $\tau\in T_0^*$. 
\item $\tau^*\in T_0^*$ and 
$W_0(\tau^*)\geq \overline{W}_+(\tau)$ for $\tau\in T_+^*$. 
\end{itemize}
	
\end{definition}

We will identify the optimal precision in the above sense by means of four parameters: the coefficient of dispersion $\zeta$, the coefficient of volatility $\eta$, the slope of the best response $\alpha$, and the coefficient of information costs $\lambda$.

\section{The welfare effect of public information} \label{The welfare effect of public information}

In this section, we study the welfare effect of public information and obtain the optimal precision 
under information acquisition  $T_+^*$ and that under no information acquisition $T_0^*$. 
We assume that $\tau_\theta$ is sufficiently small, i.e.,  $\tau_\theta<f(0)$. 

First, we consider the case of no information acquisition, i.e., $\gamma=0$ and $\tau\geq f(0)$.
The welfare is $W_0(\tau)$ by Lemma \ref{welfare lemma}, which is strictly increasing (decreasing) if the volatility's coefficient is strictly positive (negative) because 
${\partial W_0(\tau)}/{\partial \tau}={\eta \beta^2\tau^{-2}}/{(1-\alpha)^2}$.  
Therefore, we have the following proposition.
\begin{proposition}\label{T0}	
It holds that 
\begin{equation}
T_0^*=
\begin{cases}
\{f(0)\} &\text{ if }\eta<0,\\
[f(0),\infty] &\text{ if }\eta=0,\\
\{\infty\} &\text{ if }\eta>0.
\end{cases}
\end{equation}	
\end{proposition}

Next, we consider the case of information acquisition, i.e., $\gamma=f(\gamma)$. 
The welfare is $\overline{W}_+(\tau)$ by Lemma \ref{welfare lemma}. 
Recall that the domain of $\overline{W}_+(\tau)$ is $[\tau_\theta,\bar\tau]$, and the corresponding range of the information fraction is 
\begin{equation}
\{\gamma\in [0,1)\mid \tau=f(\gamma),\ \tau\in [\tau_\theta,\bar\tau]\}=
\{\gamma\in [0,1)\mid f(\gamma)\geq \tau_\theta\}=[0,\overline\phi(\tau_\theta)]
\notag\label{range of gamma}
\end{equation}
because $\tau_\theta<f(0)$ is assumed. 
Note that, for each $\tau_+^*\in T_+^*$, 
\[
\overline{W}_+(\tau_+^*)=\max_{\tau\in [\tau_\theta,\bar\tau]}\overline{W}_+(\tau)=\max_{\tau\in [\tau_\theta,\bar\tau]}\max_{\gamma\in\Gamma(\tau)}{W}_+(\gamma)=\max_{\gamma\in [0,\overline\phi(\tau_\theta)]}{W}_+(\gamma).\]  
Thus, to obtain $T_+^*$, it is enough to solve the last optimization problem, and its solution is given by the following lemma, which is an immediate consequence of Lemma \ref{welfare lemma}.
\begin{lemma}\label{gamma star lemma}
For each $\gamma\in [0,1)$, it holds that 
\begin{align}
\frac{dW_+(\gamma)}{d\gamma} =
 \frac{\lambda}{2}\left(\zeta -\eta\frac{1-2 \alpha }{ (1-\alpha)^2}-\frac{1}{1-\gamma}\right)\gtrless 0 \ \Leftrightarrow\ \gamma\gtrless \gamma^*_+,
\notag\label{gamma key eq}	
\end{align}
where 
\begin{equation}
	\gamma^*_+=\gamma^*_+(\zeta,\eta,\alpha)\equiv 
\begin{cases}
\displaystyle 1-\frac{1}{ \zeta -(1-2 \alpha ) \eta /(1-\alpha)^2} 
&\text{if }\zeta -(1-2 \alpha) \eta /(1-\alpha)^2>1,\\
0 \displaystyle & \text{otheriwse.} 	
\end{cases}
\label{def delta}
\end{equation}
Thus, $\gamma^*_+$ uniquely maximizes $W_+(\gamma)$ over $[0,1)$. 	
\end{lemma}

Let $\tau^*_+\equiv f(\gamma^*_+)$. 
Then, Lemma \ref{gamma star lemma} implies that 
$\tau^*$ uniquely maximixes $\overline{W}_+(\tau)$ over $[\tau_\theta,\bar\tau]$ 
if 
$ \tau_\theta<\tau^*_+$.\footnote{If $ \tau_\theta\geq \tau^*_+$, then $\overline{W}_+(\tau_\theta)>\overline{W}_+(\tau)$ 
for all $\tau>\tau_\theta$, but we focus on the case with sufficiently small $\tau_\theta$ to make the discussion simple.} 

\begin{proposition}\label{main proposition 1}
Assume that $\tau_\theta<\tau_+^*= f(\gamma^*_+)$.
Then, $T_+^*=\{\tau_+^*\}$. 
\end{proposition}

By Proposition \ref{main proposition 1}, no disclosure ($\tau=\tau_\theta$) is suboptimal for sufficiently small $\tau_\theta$. 
When the policymaker discloses no information, the agents acquire a substantial amount of information to reduce uncertainty, incurring inefficiently large costs.  
In fact,  as $\tau$ approaches zero,   
the cost goes to infinity, so the expected welfare goes to minus infinity: 
\[
\lim_{\tau\to 0}\overline{W}_+(\tau)=\lim_{\tau\to 0}{W}_+(\overline\phi(\tau))
=\lim_{\gamma\to 1}{W}_+(\gamma)=\zeta D_+(1)+\eta V_+(1)-\lambda/2\cdot\lim_{\gamma\to 1}\log(1-\gamma)^{-1}=-\infty.
\]



Using Lemma \ref{gamma star lemma}, we study the welfare effect of public information and gain intuition about the optimal precision $\tau_+^*$.
We focus on the case of $\tau<f(0)$ and $\Gamma(\tau)=\{\overline{\phi}(\tau)\}$, i.e., the second-period equilibrium is unique.  
Because $\overline{W}_+(\tau)=W_+(\overline{\phi}(\tau))$, we obtain 
\begin{align}
	\frac{d \overline{W}_+(\tau)}{d \tau}
	&=
\zeta\frac{d D_+(\overline{\phi}(\tau))}{d \tau}
+\eta\frac{d V_+(\overline{\phi}(\tau))}{d \tau}
-\lambda\frac{d\log (1-\overline{\phi}(\tau))^{-1}}{d\tau}\notag\\
&=
\left(\zeta -(1-2 \alpha)\eta/ (1-\alpha)^2- 1/(1-\overline{\phi}(\tau))\right)
\cdot {\lambda\overline{\phi}'(\tau)}/{2}. \label{key eq welfare 0}
\end{align}
This leads us to the following proposition, which provides a necessary and sufficient condition for welfare to increase with public information.

\begin{proposition}\label{main proposition 3}
If $\tau<f(0)$, it holds that 
\begin{align}
\frac{\partial \overline{W}_+(\tau)}{\partial \tau}\gtrless 0&\ \Leftrightarrow\ 
\zeta -(1-2 \alpha)\eta/ (1-\alpha)^2- 1/(1-\overline{\phi}(\tau))\lessgtr 0,\label{key eq welfare}\\ 
\frac{\partial \overline{W}_+(\tau)}{\partial \tau}< 0&\ \Rightarrow\ 
\zeta -(1-2 \alpha)\eta/ (1-\alpha)^2> 1.\label{key eq welfare'}
\end{align}
\end{proposition}
The change in welfare, $\partial \overline{W}_+/\partial \tau$, consists of those in the dispersion term, the volatility term, and the cost term, which correspond to the three terms in \eqref{key eq welfare 0} and \eqref{key eq welfare}, respectively. 	 
The last term in \eqref{key eq welfare}, $1/(1-\overline{\phi}(\tau))$, measures the cost reduction induced by the crowding-out effect, which is decreasing in $\tau$ and approaches infinity as $\tau$ goes to zero (becuase $\overline{\phi}(\tau)$ goes to one). 
This is why welfare increases with public information when public information is sufficiently imprecise.  
Note that the cost reduction is small when $\tau$ is large and the agents acquire a small amount of information.

On the other hand, welfare can decrease with public information only if $\zeta -(1-2 \alpha ) \eta /(1-\alpha)^2>1$ by \eqref{key eq welfare'} (see Figure \ref{fig1}), which is equivalent to $\gamma_+^*>0$ by \eqref{def delta}. 
That is, 
when the coefficient of dispersion $\zeta$ is large enough, more precise public information can be harmful. 
We can understand this by the fact that the dispersion is decreasing in $\tau$: 
\begin{equation}
\frac{d D_+(\overline{\phi}(\tau))}{d \tau}={\lambda\overline{\phi}'(\tau)}/{2}<0.	\notag \label{dispersion decreases}
\end{equation}
This is because the dispersion equals the difference between the variance of an action and that of the average action, and more precise public information brings them closer. 
In contrast, when the coefficient of volatility $\eta$ is large enough, more precise public information is always beneficial if $\alpha<1/2$, yet it can be harmful if $\alpha>1/2$.  
This can be understood  by the fact that the volatility is increasing in $\tau$ if and only if $\alpha<1/2$: 
\[
\frac{d V_+(\overline{\phi}(\tau))}{d \tau}=-{\lambda\overline{\phi}'(\tau)}/2\cdot (1-2\alpha)/(1-\alpha)^2\gtrless 0\ \Leftrightarrow \ \alpha \lessgtr 1/2.
\]
The volatility equals the covariance of actions since $\var[A]=\cov[a_i,A]=\cov[a_i,a_j]$, which equals the corelation coefficient multiplied by the variance of an individual action. 
More precise public information increases the correlation coefficient but decreases the variance through the crowding-out effect. 
When $\alpha>1/2$, the crowding-out effect is so strong and a decrease in the variance is so large that the covariance decreases with public information. 


\section{The optimal disclosure rule}\label{The optimal precision}
Using the result in Section \ref{The welfare effect of public information}, we obtain the optimal precision of public information. 
When $\eta\neq 0$, there are three candidates for the optimal precision, $\tau_+^*$, $f(0)$, and $\infty$, by 
Propositions \ref{T0} and \ref{main proposition 1}. 
However, we can focus on $\tau_+^*$ and $\infty$ for the following reason: if $\gamma_+^*=0$, then $\tau_+^*=f(0)$; if $\gamma_+^*>0$, then $\tau_+^*=f(\gamma_+^*)\neq f(0)$, and 
$\overline{W}_+(\tau_+^*)>\overline{W}_+(f(0))\geq W_0(f(0))$. 
Therefore, the optimal precision is $\tau_+^*$ if $\overline{W}_+(\tau_+^*)-W_0(\infty)>0$ and $\infty$ if $\overline{W}_+(\tau_+^*)-W_0(\infty)<0$. 
This observation leads us to the following proposition. 

\begin{proposition}\label{main proposition opt}
Assume that $\tau_\theta<\tau_+^*= f(\gamma^*_+)$.
The optimal precision $\tau^*$ is uniquely given by
\[
\tau^*=
\begin{cases}
	\infty & \text{ if $\chi<0$ and $\eta>0$},\\
	\tau_+^* & \text{ if $\chi>0$ or $\eta<0$},
\end{cases}
\]
where $\chi\equiv \zeta-1-2 \eta /(1-\alpha)-\log(1-\gamma_+^*)^{-1}$. 
Both $\infty$ and $\tau_+^*$ are optimal in the other nongeneric case ($\chi=0$ or $\eta=0$).
\end{proposition}

Figure \ref{fig:flexible} illustrates the optimal precision on the $\zeta\eta$-plane.  
Full disclosure is optimal if $(\zeta,\eta)$ is in the upper region, where the coefficient of volatility is positive and large enough; partial disclosure is optimal if $(\zeta,\eta)$ is in the lower region. 
If $\alpha\leq 1/2$, the volatility necessarily increases with public information, but if $\alpha>1/2$, it can decrease due to the crowding-out effect. 
Even in the latter case, when public information is sufficiently precise, the volatility necessarily increases with public information because the agents do not acquire private information and the crowding-out effect disappears. 
This is why a sufficiently large coefficient of volatility guarantees the optimality of full disclosure. 

No disclosure is never optimal because public information reduces the cost of information, which is in sharp contrast to the case of exogenous private information. 
To see this, assume that an agent receives a private signal $\tilde\theta_i$ that follows a fixed normal distribution in the second period at no cost \citep{angeletospavan2007}. A private signal is conditionally independent across agents given $\theta$ and the precision of private information (defined as $\tau_i\equiv 1/\var[\theta|\tilde\theta_i]-\tau_\theta$) is exogenously fixed. 
\citet{uiyoshizawa2015} study optimal disclosure of public information in this case and show that no disclosure is optimal if $\eta$ is small enough, as illustrated in Figure \ref{fig:exogenous} and formally stated in the next proposition.\footnote{See Corollary 5 of \citet{uiyoshizawa2015}.}

\begin{proposition}[Ui and Yoshizawa, 2015]\label{Proposition UY}
Assume that the precision of private information is exogenously fixed. 
For any precision of private information, the optimal precision $\tau^*$ is uniquely given by 
\[
\tau^*=
\begin{cases}
\infty&\text{ if $\eta> \max\{0, (1-\alpha)\zeta/2\}$},\\
0 &\text{ if $\eta< \min  \{0, 2(1-\alpha)\zeta/3\}$}.
\end{cases}
\]	
In the other case, $\tau^*$ depends upon the precision of private information. 
\end{proposition}

Even if the precision of private information is endogenously determined, no disclosure can be optimal as long as the cost is convex in the precision. 
To be more specific, 
assume that an agent chooses the precision of private information before receiving a private signal and that its cost is a convex function \citep{colombofemminis2014}. 
Information acquisition is rigid in the sense that the distribution is restricted to be normal. 
\citet{ui2022} studies optimal disclosure in this case and finds that no disclosure is optimal if $\eta$ is small enough.  

In the model of rigid information acquisition, more precise public information also reduces the cost of information, but the reduction cannot be substantial enough to make no disclosure suboptimal. 
To see why, recall that the cost of information is assumed to be convex in the precision of private information. 
This implies that the precision is bounded above because the marginal cost is increasing and the marginal benefit approaches zero as the precision goes to infinity. 
Consequently, the cost of private information cannot be so substantial. 
In contrast, the precision and the cost of private information are unbounded in the model of flexible information acquisition, thus making no disclosure suboptimal.



\section{Applications}\label{Cournot}


\subsection{Cournot and investment games}\label{Cournot games}

Suppose that $h(a,\theta)=0$ in \eqref{payoff function cont}, and let $v(a,\theta)=\int u_i(a,\theta)di$ be the aggregate payoff. 
Then, it can be readily shown that $\E[v(a,\theta)]=\var[a_i]$; that is, $\zeta=\eta=1$ in \eqref{welfare 1}. Thus, the expected welfare equals the variance of an individual action minus the cost of information. 

A Cournot game \citep{vives1988} is a special case when $\alpha<0$.
Firm $i$ produces $a_i$ units of a homogeneous product at a quadratic cost $a_i^2/2$. 
An inverse demand function is $\theta-\delta \int a_jdj$,  
where $\delta>0$ is constant and $\theta$ is normally distributed. 
Then, firm $i$'s gross profit excluding the cost of information is 
\begin{equation}
(\theta-\delta A)a_i- a_i^2/2, \notag
\end{equation}
which is reduced to \eqref{payoff function cont} with $\alpha=-\delta$ and $\beta=1$ by normalization.

An investment game \citep{angeletospavan2004} is also a special case when $\alpha>0$. 
Firm $i$ chooses an investment level $a_i$ at a quadratic cost $a_i^2/2$ and receives a return $(r A+(1-r)\theta)a_i$, where $r\in (0,1)$ is constant. 
Then, firm $i$'s gross profit excluding the cost of information is 
\begin{equation}
(r A+(1-r)\theta)a_i- a_i^2/2, \notag
\end{equation}
which is reduced to \eqref{payoff function cont 1} with $\alpha=r$ and $\beta=1-r$ by normalization.

To provide a benchmark for our result, assume that private information is exogenous. 
In a Cournot game, the total profit can decrease with public information, and no disclosure can be firm optimal \citep{morrisbergemann2013}. In an investment game, the total profit necessarily increases with public information \citep{angeletospavan2004}. 
These results are formally stated as follows \citep{uiyoshizawa2015}. 

\begin{proposition}
Assume that the precision of private information is exogenously fixed and that $\zeta=\eta=1$. 
If $\alpha\geq -1/2$, welfare necessarily increases with public information. 
If $\alpha< -1/2$, welfare can decrease with public information. 
If $\alpha< -1$, no disclosure is optimal when private information is sufficiently precise. 
\end{proposition}


In brief, more precise public information can reduce the total profit if and only if the game exhibits strong strategic {substitutability}, in which case no disclosure can be optimal (see Figure~\ref{exogenous graph}). 
This result is essentially the same as in the case of rigid information acquisition \citep{ui2022}.\footnote{Under rigid information acquisition with strictly convex information costs (strictly convex in the precision of private information), more precise public information can reduce the total profit if and only if the game exhibits strong strategic {substitutability}, where the slope of the best response is smaller than in the case of exogenous private information.} 


\begin{figure} 
  \centering
  \subfloat[$\alpha <-1$.]{\label{fig: 2exv}
    \includegraphics[width=4cm, bb=0 0 675 419]{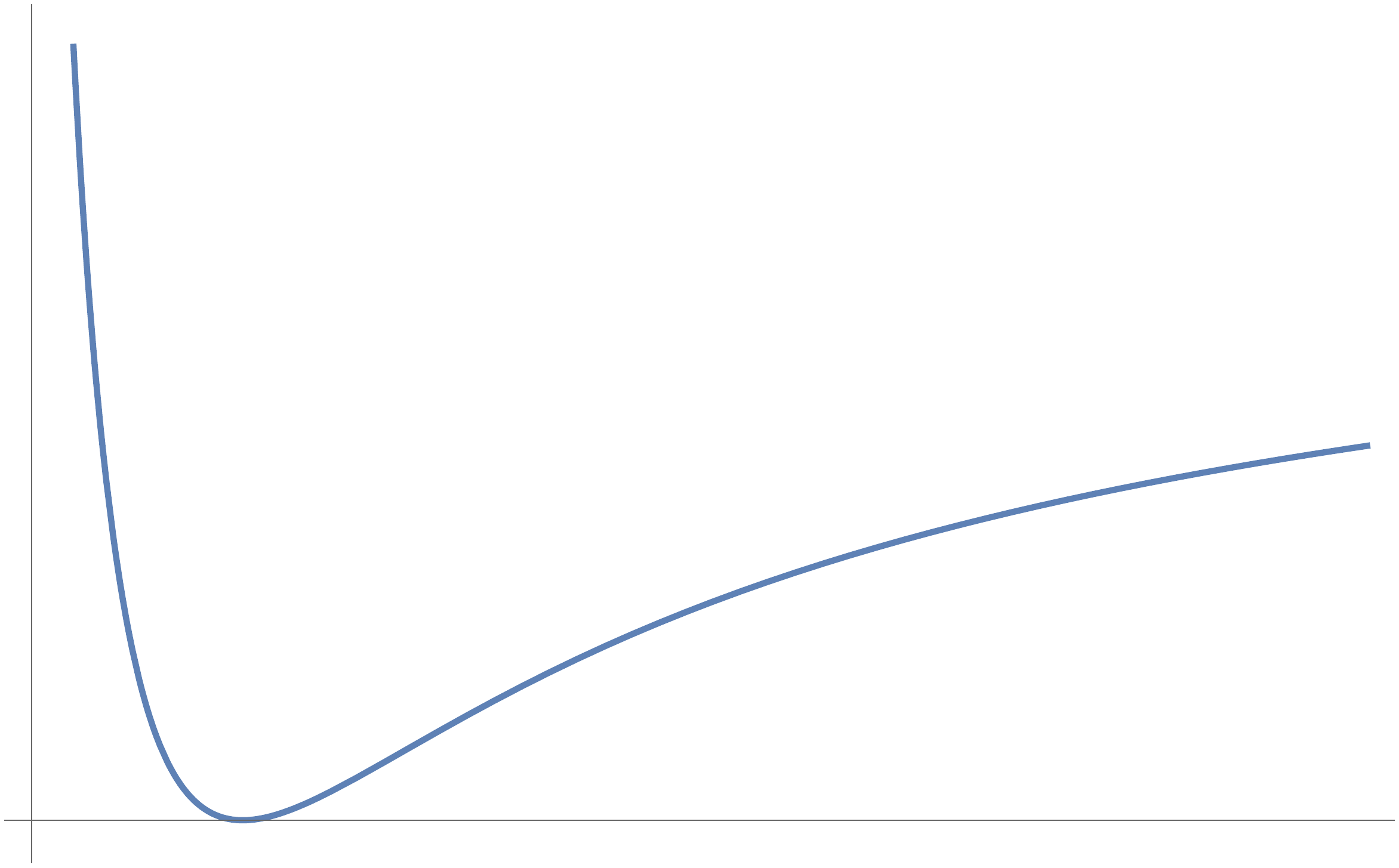}}
\qquad \ 
  \subfloat[$\alpha > 1/2$.]{\label{fig: 34exv}
    \includegraphics[width=4cm, bb=0 0 675 419]{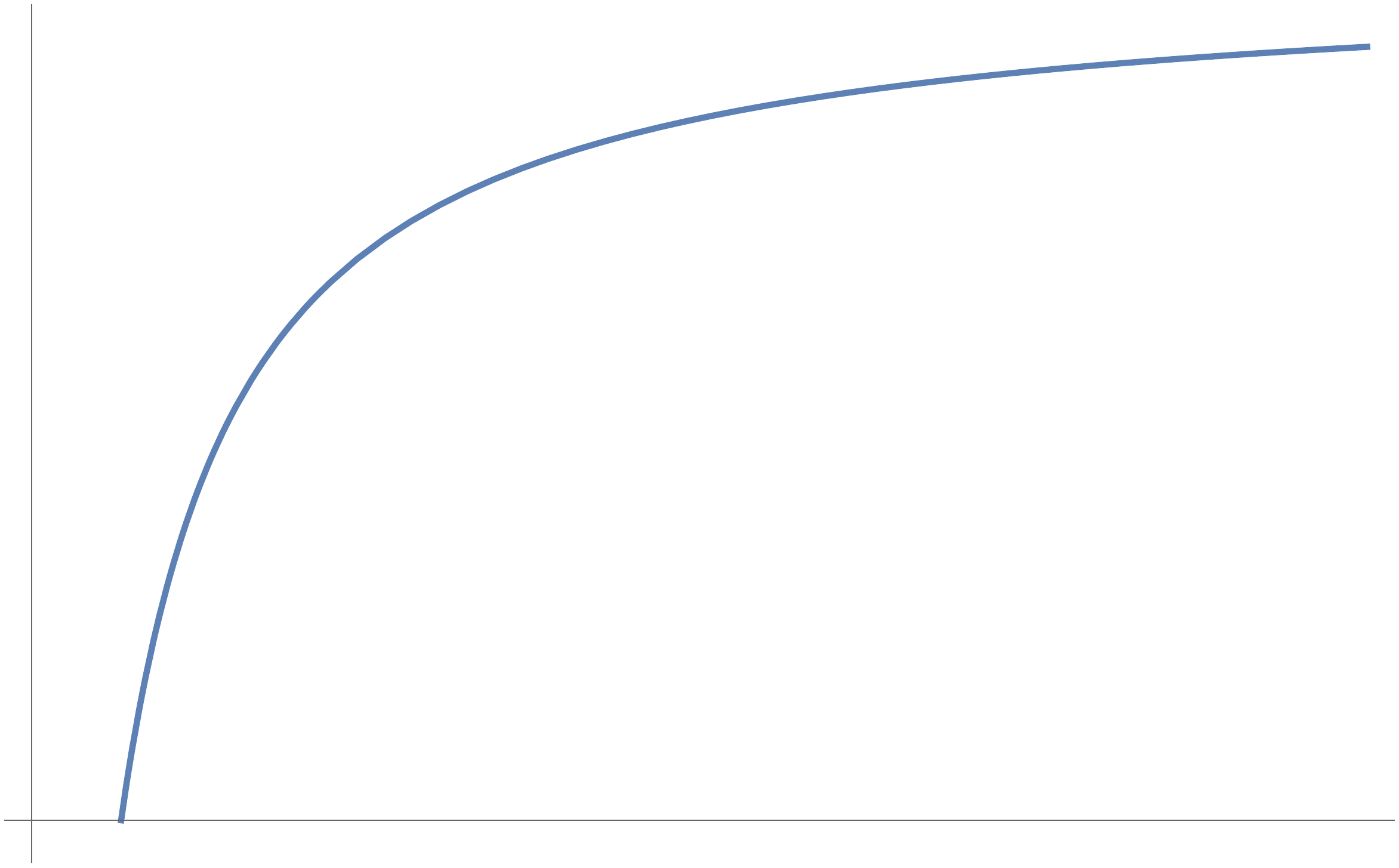}}
  \caption{The total profit ($=$ the variance) as a function of $\tau$ in the case of exogenous private information.} 
  \label{exogenous graph}
\end{figure}

We provide a contrasting result in the case of flexible information acquisition as a corollary of Propositions \ref{main proposition 3} and \ref{main proposition opt}. 
That is, more precise public information can reduce the total profit if and only if the game exhibits strong strategic {\em complementarity}, while full disclosure is always optimal. 
In other words, the total profit necessarily increases with public information in a Cournot game, whereas it can decrease with public information in an investment game. 

\begin{corollary}\label{cournot proposition}
Assume that $\zeta=\eta=1$.  
If $\alpha\leq 1/2$, then $\overline{W}_+(\tau)$ is increasing. 
If $\alpha>1/2$, then $\overline{W}_+(\tau)$ is increasing for $\tau< f(0)$ and decreasing for $\tau\in (f(0),\bar\tau)$.  
In both cases, full disclosure is optimal. 
\end{corollary}


\begin{figure} 
  \centering
  \subfloat[The gross profit ($\alpha <-1$).]{\label{fig: 2v}
    \includegraphics[width=4cm, bb=0 0 675 419]{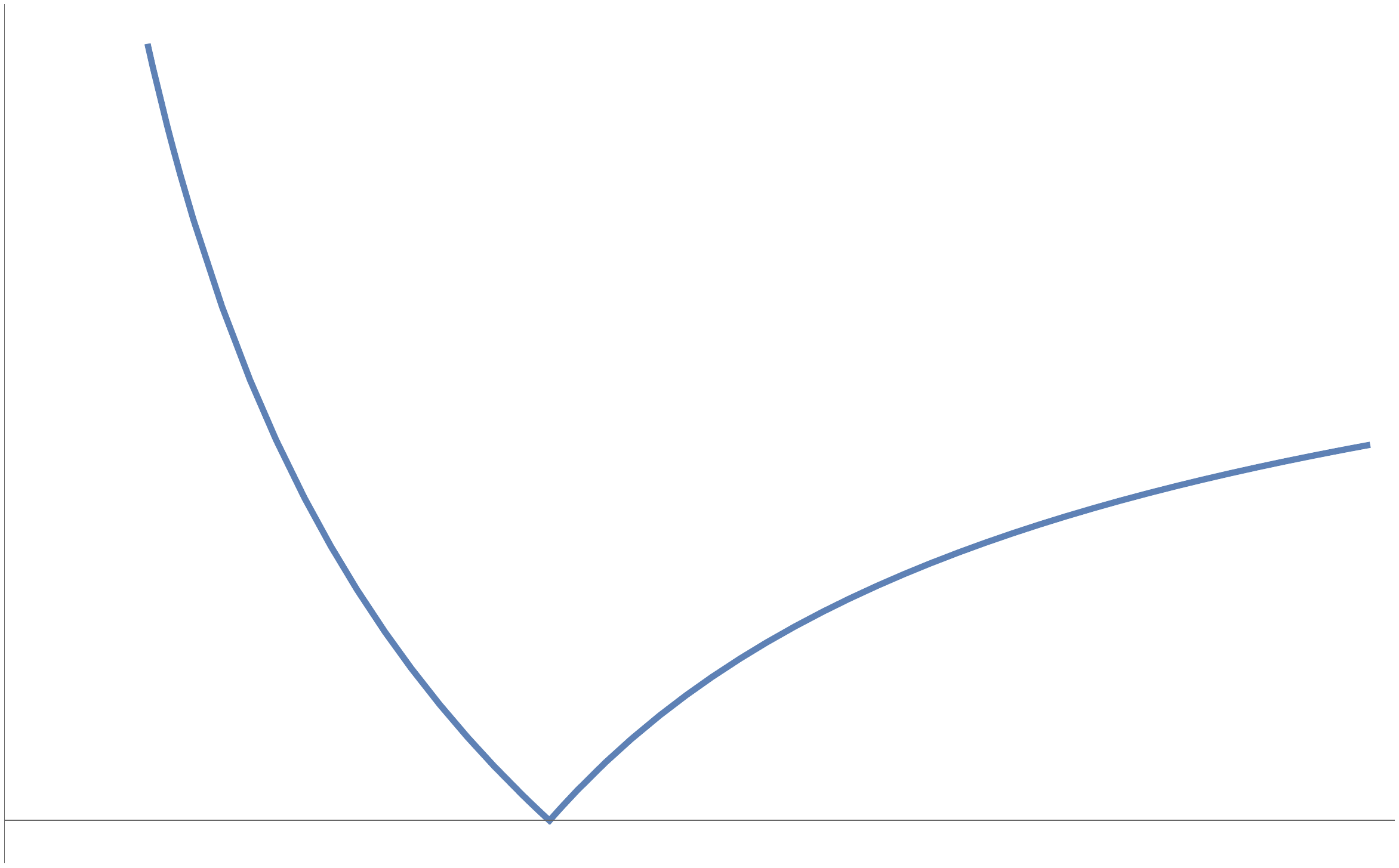}}
\qquad \ 
  \subfloat[The gross profit ($\alpha>1/2)$.]{\label{fig: 34v}
    \includegraphics[width=4cm, bb=0 0 675 419]{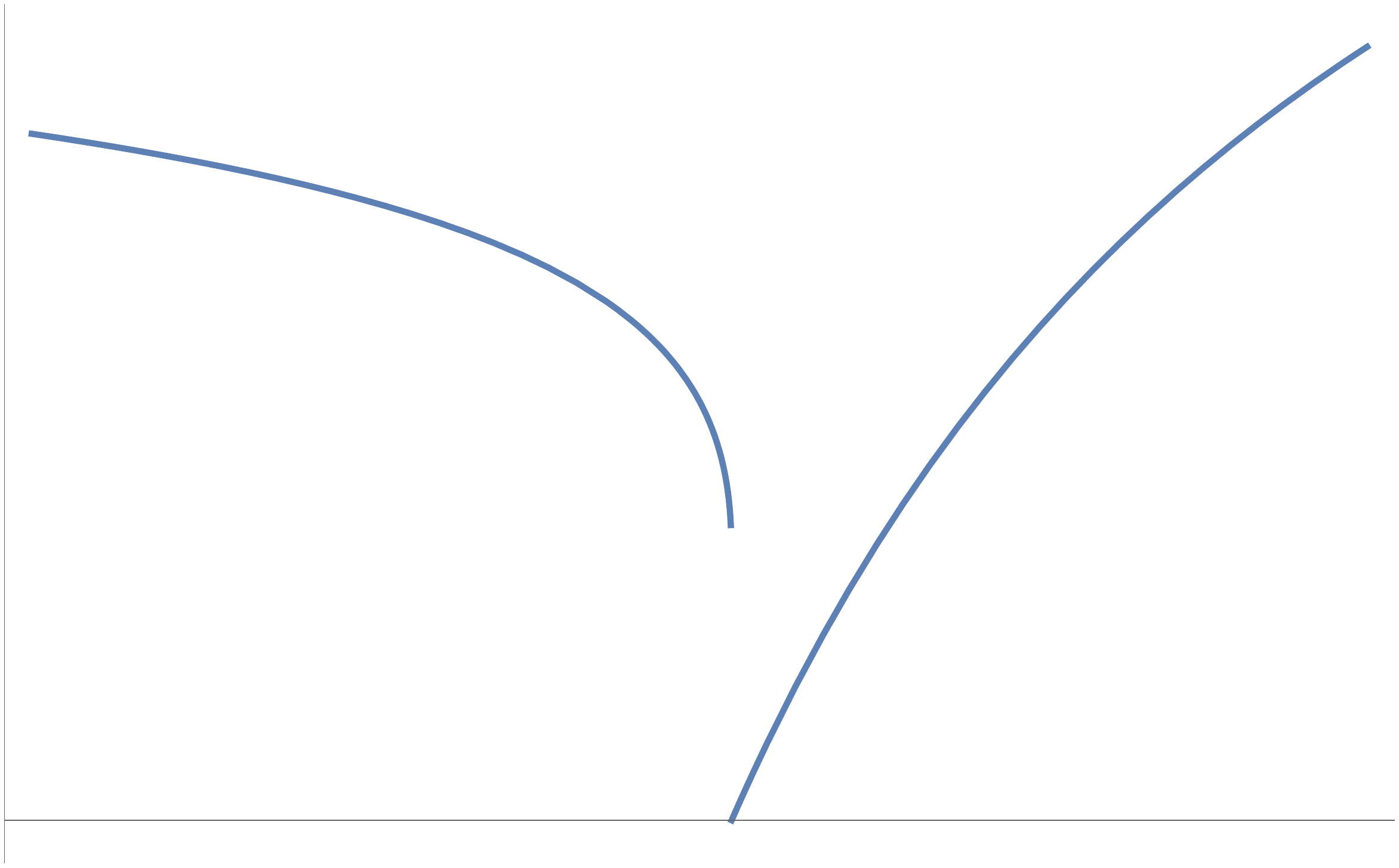}}\\
  \subfloat[The total profit ($\alpha <-1$).]{\label{fig: 2w}
    \includegraphics[width=4cm, bb=0 0 675 419]{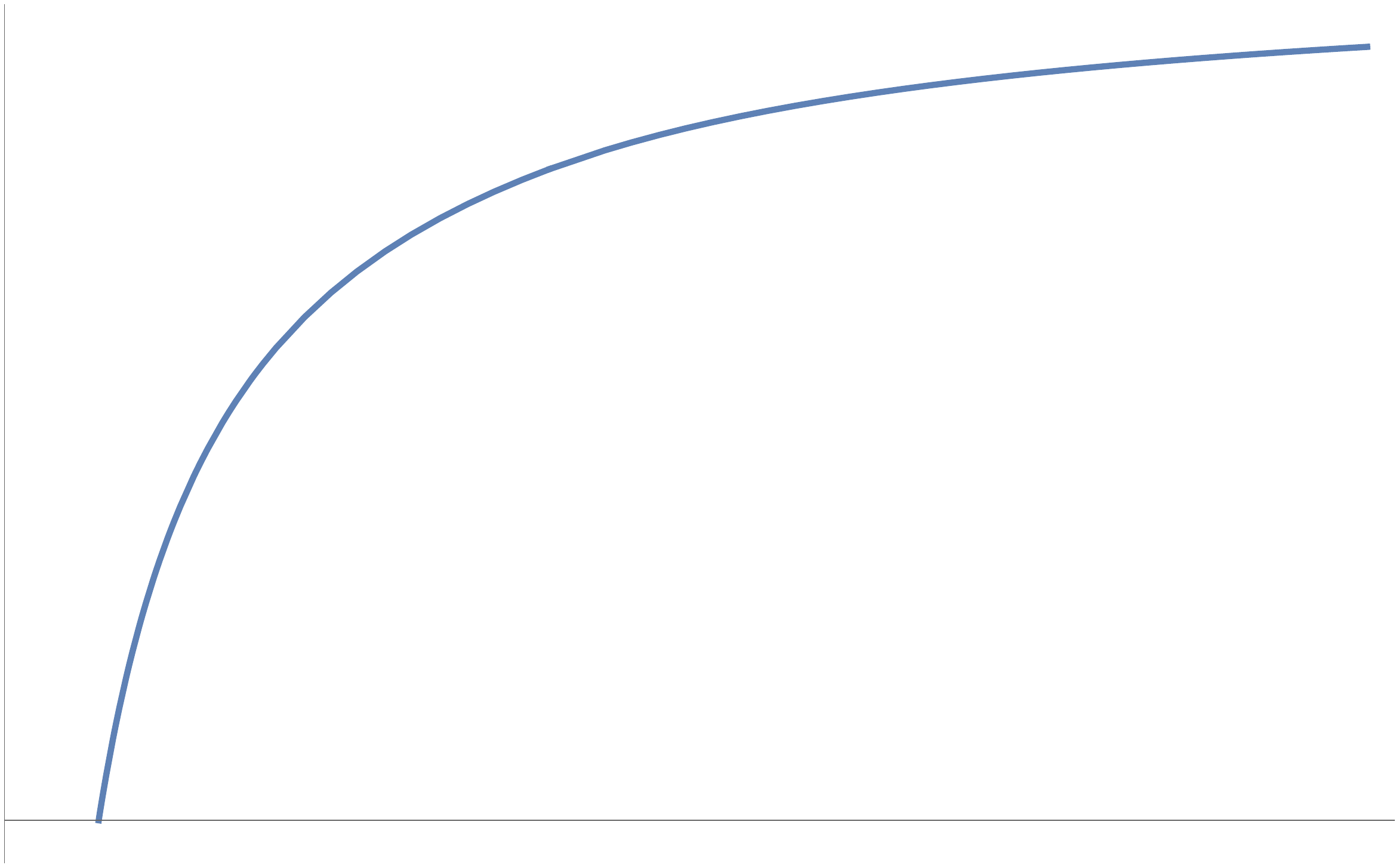}}
\qquad \ 
  \subfloat[The total profit ($\alpha > 1/2$).]{\label{fig: 34w}
    \includegraphics[width=4cm, bb=0 0 675 419]{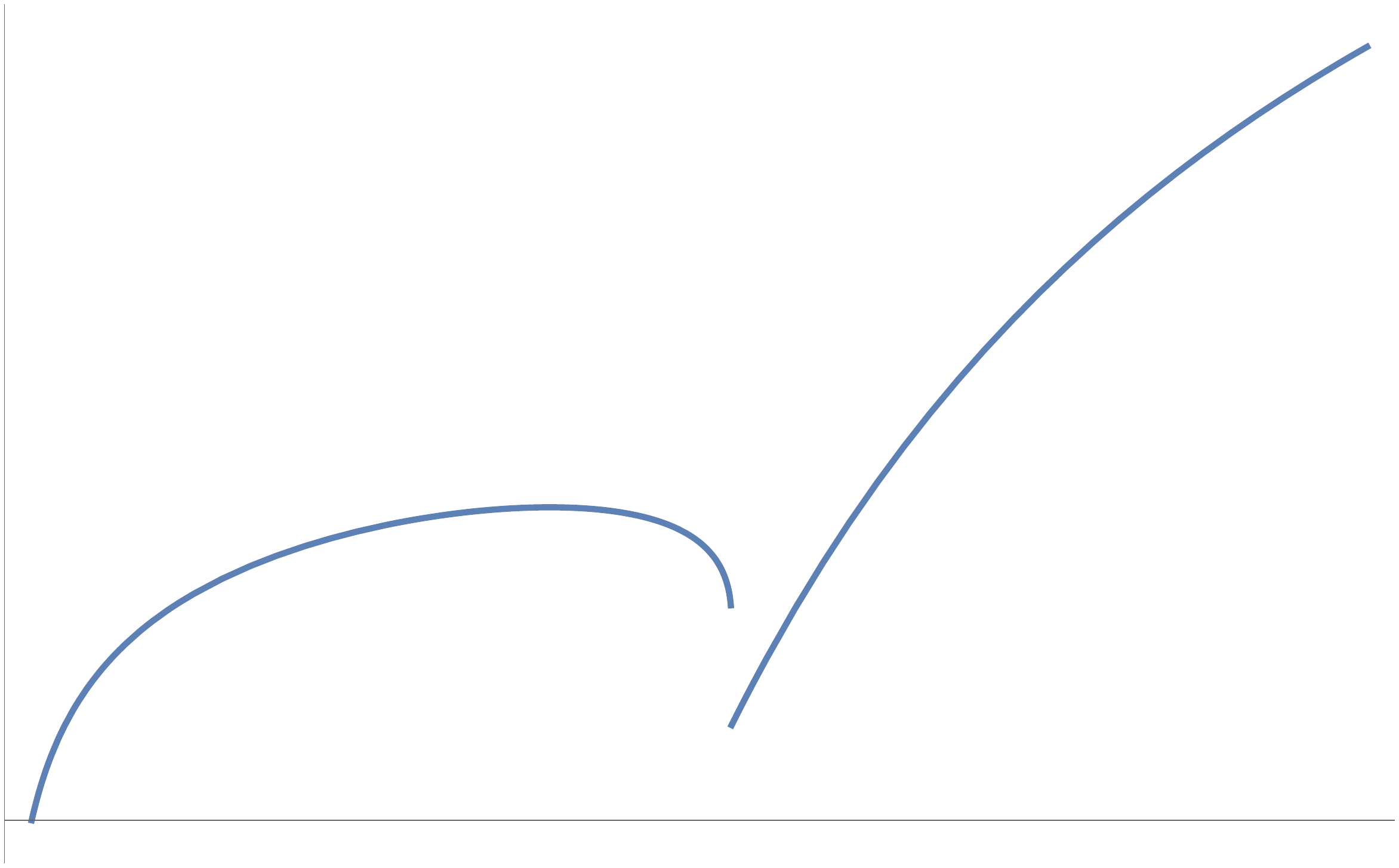}}
  \caption{The gross profit and total profit as functions of $\tau$. In the case of multiple equilibria, the maximum value is plotted. When $\alpha>1/2$, the gross profit is disconstinuous at $\tau=\bar\tau$.} 
  \label{flexible graph}
\end{figure}

The difference in the above results arises from the crowding-out effect enhanced by strategic complementarity. 
Recall that the total profit equals the gross profit 
minus the cost of information. 
When the policymaker provides more precise public information,  the gross profit decreases if $\tau<\bar\tau$ 
and increases if $\tau>\bar\tau$ (see Figures \ref{fig: 2v} and \ref{fig: 34v}), while the cost of information decreases if $\tau<\bar\tau$ and vanishes if $\tau>\bar\tau$. 
This is because the crowding-out effect reduces private information if $\tau<\bar\tau$ and disappears if $\tau>\bar\tau$.  
Since the cost reduction is substantial when $\tau$ is very small, 
the total profit increases when $\tau$ is small enough as well as large enough, thus making full disclosure optimal (see Figures \ref{fig: 2w} and \ref{fig: 34w}).   
However, when the degree of strategic complementarity is sufficiently strong (i.e., $\alpha> 1/2$) and the cost reduction is modest (i.e., $\tau$ is close to $\bar \tau$), a decrease in the gross profit is so substantial that the total profit decreases (see Figure \ref{fig: 34w}).

\subsection{Beauty contest games}\label{Beauty contest games}

We consider a beauty contest game \citep{morrisshin2002}, where 
$\alpha =r\in(0,1)$, $\beta=1-r$, and $v(a,\theta)-(a_i-\theta)^2$. 
An agent's target is the weighted mean of the state and the aggregate action, $(1-r) \theta+rA$. 
The welfare is the negative of the mean squared error of an individual action from the state minus the cost of information, $-\E[(a_i-\theta)^2]-C(a_i)$. 
Because the welfare has a negative value, full disclosure attains the maximum welfare and is optimal. 
It can be readily shown that $\zeta=1+r$ and $\eta=1-r$ in \eqref{welfare 1}.

In the case of exogenous private information, the following result is well known: if the degree of strategic complementarity is strong enough, more precise public information can be harmful to welfare. 
\begin{proposition}[Morris and Shin, 2002]\label{cournot proposition  exogenous}
Assume that the precision of private information is exogenously fixed  
and that $(\zeta,\eta)=(1+r,1-r)$. 
If $r>1/2$, welfare can decrease with public information. 
\end{proposition}

The harmful effect arises from the fact that agents place too much weight on public information, which is entailed by a coordination motive under strong strategic complementarity, and overreact to public information. 


In the case of flexible information acquisition,  welfare can decrease with public information even if $r<1/2$; that is, welfare is more likely to decrease. 


\begin{corollary}\label{beauty contest proposition}
Assume that $\alpha=r$ and $(\zeta,\eta)=(1+r,1-r)$. 
If $r> (3-\sqrt{5})/2\simeq 0.38$, then $\overline{W}_+(\tau)$ is deceasing for $\tau \in (\tau_+^*,\bar\tau)$, where $\tau_+^*=f({(r^2-3 r +1)}/{(r(r-2))})$. 
\end{corollary}

\begin{figure} 
  \centering
  \subfloat[The gross welfare is decreasing whenever the agents acquire private information.]{\label{fig: bcvend}
    \includegraphics[width=4cm, bb=0 0 675 419]{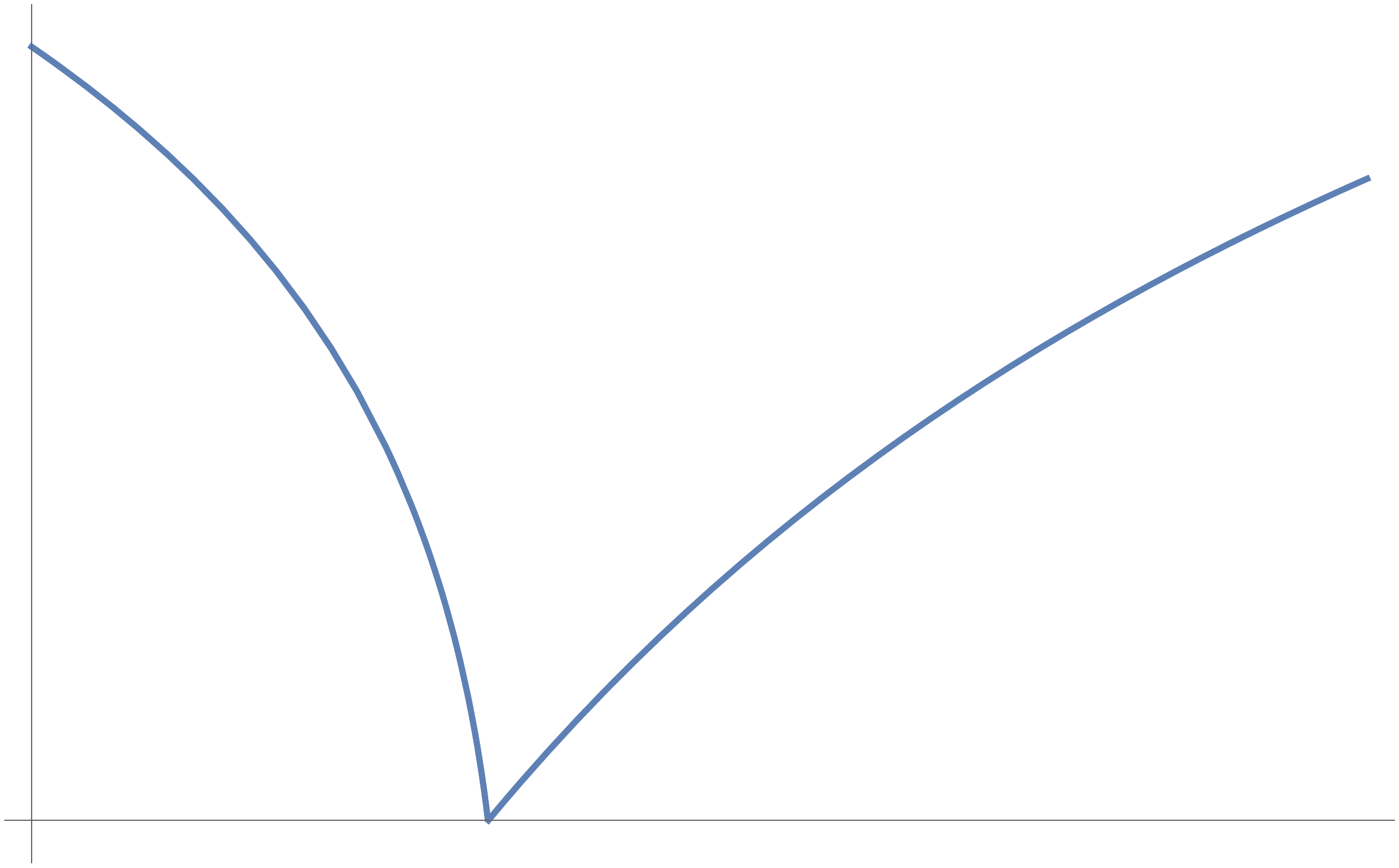}}
    \qquad \ 
  \subfloat[The net welfare is decreasing when the agents acquire a small amount of private information.]{\label{fig: bcend}
    \includegraphics[width=4cm, bb=0 0 675 419]{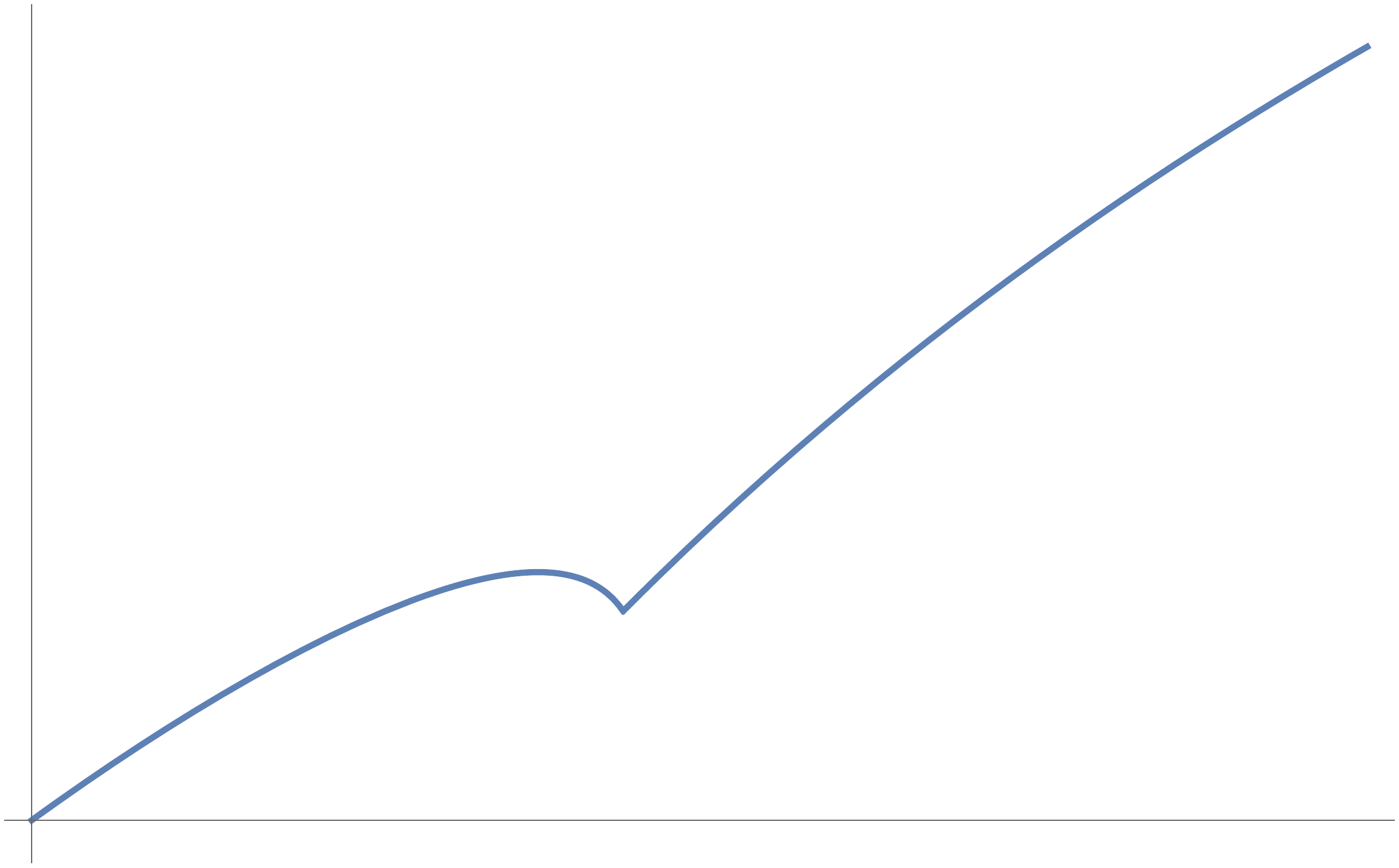}}
  \caption{Welfare in a beauty contest game ($\alpha>(3-\sqrt{5})/2)$).} 
  \label{beauty contest graph}
\end{figure}

This result is attributed to the crowding-out effect enhanced by strategic complementarity, which is essentially the same as the case of an investment game. 
When the policymaker provides more precise public information, the gross welfare excluding the cost of information decreases if $\tau<\bar\tau$ due to the crowding-out effect (see Figures \ref{fig: bcvend}). 
Nonetheless, the net welfare increases when $\tau$ is small enough because of a substantial decrease in the cost of information. 
However, when the degree of strategic complementarity is sufficiently strong (i.e., $r> (3-\sqrt{5})/2)$) and the cost reduction is modest (i.e., $\tau$ is close to $\bar \tau$), a decrease in the gross welfare is so substantial that the net welfare decreases.

\begin{appendices}
	
\bigskip
\appendix
\setcounter{section}{0}
\setcounter{theorem}{0}
\setcounter{lemma}{0}
\setcounter{claim}{0}
\setcounter{proposition}{0}
\setcounter{definition}{0}

\renewcommand{\theequation}{A.\arabic{equation} }
\setcounter{equation}{0}

\renewcommand{\thetheorem}{\Alph{theorem}}
\renewcommand{\thelemma}{\Alph{lemma}}
\renewcommand{\theclaim}{\Alph{claim}}
\renewcommand{\theproposition}{\Alph{proposition}}
\renewcommand{\thedefinition}{\Alph{definition}}

\section{Proofs for Section \ref{section: The second-period subgame}}\label{Proofs for Section {section: The second-period subgame}}

\begin{proof}[Proof of Proposition \ref{proposition 1}]
Assume that the distribution of $(a_i,A,\theta)$ is an equilibrium. 
Because $\Et[\alpha A+\beta\theta|a_i]=a_i$, 
we have $\Et[\alpha A+\beta\theta]=\alpha\Et[ A]+\beta\tilde\theta=\Et[a_i]=\Et[A]$. 
By solving this for $\Et[A]$, we obtain $\Et[a_i]=\Et[A]=\beta \tilde\theta/(1-\alpha)$.

Let $x=\alpha A+\beta\theta$ and $y=a_i$. 
Then, $y$ is an optimal solution of \eqref{CT lemma} and satisifes the condition in Lemma \ref{rate distortion lemma}.
Suppose $\lambda/2<\vart[\alpha A+\beta\theta]$. 
By Lemma \ref{rate distortion lemma}, 
\[
\vart[a_i]=\vart[\alpha A+\beta\theta]-\lambda/2=\vart[a_i]/\gamma-\lambda/2,
\]
which implies $\vart[a_i]={\lambda \gamma}/({2(1-\gamma)})$, 
i.e., \eqref{main equi 1}. 

We obtain a condition for the existence of $\gamma\in (0,1)$. 
An agent acquires information about $\alpha A+\beta \theta$ but does not pay attention to $A$ and $\theta$ separately (because it is more costly).  
Thus, we must have 
\begin{equation}
	\Et[a_i|\alpha A+\beta \theta]=\Et[a_i|A,\theta]. \label{key calc eq}
\end{equation} 
Using the formula for conditional distributions, we obtain 
\begin{align}
\Et[a_i|\alpha A+\beta \theta]
&=\Et[a_i]+\frac{\covt[a_i,\alpha A+\beta \theta]}{\vart[\alpha A+\beta \theta]}(\alpha A+\beta \theta-\Et[\alpha A+\beta \theta])\notag\\
&=\Et[a_i]+\gamma(\alpha A+\beta \theta-\Et[\alpha A+\beta \theta])\label{key calc eq'}
\end{align}
because 
$
\covt[a_i,\alpha A+\beta \theta]
=\covt[\Et[\alpha A+\beta \theta|a_i],\alpha A+\beta \theta]=
\vart[\Et[\alpha A+\beta \theta|a_i]]=\vart[a_i]$.
Similarly, 
\begin{align}
\Et[a_i|A, \theta]
&=\Et[a_i]+
\begin{pmatrix}
\covt[a_i,A] & \covt[a_i,\theta]
\end{pmatrix}
\begin{pmatrix}
\vart[A] & \covt[A,\theta]\\
\covt[\theta,A] & \vart[\theta]
\end{pmatrix}^{-1}
\begin{pmatrix}
A-\Et [A] \notag \\\theta-\tilde\theta
\end{pmatrix}\\
&=\Et[a_i]+(A-\Et [A])\label{key calc eq''}
\end{align}
because 
$(\covt[a_i,A],\covt[a_i,\theta])=(\vart[A],\covt[A,\theta])$. 
Hence, we have 
\begin{align}
A-\Et[A]
=\gamma(\alpha A+\beta \theta-\Et[\alpha A+\beta \theta])
\label{key calc 2}
\end{align}
by \eqref{key calc eq}, \eqref{key calc eq'}, and \eqref{key calc eq''}, and 
the variances of both sides are equal, i.e.,  
\begin{equation}
\vart[A]=\gamma^2\vart[\alpha A+\beta \theta]=\gamma\vart[a_i]=
{\lambda \gamma^2}/({2(1-\gamma)}),\label{key calc 3}	
\end{equation}
which implies \eqref{main equi 2}. 
Solving \eqref{key calc 2} for $A$, we have
$A-\Et[A]
=\beta \gamma(\theta-\tilde\theta)/(1-\alpha \gamma)$, 
which implies \eqref{main equi 0} and 
\begin{equation}
\vart[A]
=\beta^2 \gamma^2 \vart[\theta]/(1-\alpha \gamma)^2=\beta^2 \gamma^2\tau^{-1}/(1-\alpha \gamma)^2.\label{key calc 4}
\end{equation}
Then, \eqref{main result eq 2} follows from \eqref{key calc 3} and \eqref{key calc 4}.

By multiplying both sides of $\Et[\alpha A+\beta\theta|a_i]=a_i$ by $a_i$ and taking the expectation, we have 
$\alpha \covt[a_i,A]+\beta \covt[a_i,\theta]=\vart[a_i]$, which implies \eqref{main equi 3}.

In summary, if an equilibrium with $\gamma>0$ exists, then \eqref{main equi 0}, \eqref{main equi 1}, \eqref{main equi 2}, and \eqref{main equi 3} hold.  
It is easy to prove the converse: if the joint normal distribution of $(a_i,A,\theta)$ satisfies these equations, then it is an equilibrium. 
We can also verify the following. 
\begin{itemize}
	\item Suppose that $\alpha\leq 1/2$. Then, $f(\gamma)$ is decreasing for all $\gamma\in [0,1)$ (see Figure \ref{fig:case (i)}). Thus, \eqref{main result eq 2} has a unique solution in $(0,1)$ if and only if $\tau< f(0)=2\beta^2/\lambda$. 
	\item Suppose that $\alpha> 1/2$. Then, $f(\gamma)$ is increasing for $\gamma\in [0,(2\alpha-1)/\alpha]$ and decreasing for $\gamma\in [(2\alpha-1)/\alpha,1)$ (see Figure \ref{fig:case (ii)}). Thus, \eqref{main result eq 2} has a unique solution in $(0,1)$ if $\tau< f(0)$ or $\tau=f(({2 \alpha -1})/{\alpha })={\beta^2}/{(2 \alpha(1-\alpha)\lambda)}$, and two solutions if $f(0)<\tau<f(({2 \alpha -1})/{\alpha })$.
\end{itemize}

Next, we show that an equilibrium with no information acquisition exists if and only if $\tau\geq 2\beta^2/\lambda$. 
Suppose that an equilibrium with no information acquisition exists. 
Beacuse $a_i$ and $A$ are constant, we must have $\lambda/2\geq \vart[\alpha A+\beta\theta]=\beta^2/\tau$ by Lemma \ref{rate distortion lemma}. 
Conversely, suppose that $\tau\geq 2\beta^2/\lambda$.  
In an equilibrium, if $A$ is constant, then 
 $\lambda/2\geq \vart[\alpha A+\beta\theta]=\beta^2/\tau$, so $a_i$ is also constant by Lemma~\ref{rate distortion lemma}.  
Because $a_i=\Et[\alpha A+\beta\theta]$ and $a_i=\Et[a_i]=\Et[A]=A$,  we have $a_i=A=\alpha\tilde\theta/(1-\beta)$. 
\end{proof}

\section{Proofs for Section \ref{The crowding-out effect of public information}
}\label{Proofs for Section The crowding-out effect of public information}

\begin{proof}[Proof of Lemma \ref{lemma: crowding-out}]
Because $\overline\phi(\tau)$ is the unique solution to $\tau=f(\gamma)$,  
\[
\overline\phi'(\tau)=1/f'(\overline\phi(\tau))
=\frac{\lambda  (1-\alpha  \overline\phi(\tau))^3}{2 \beta ^2 ( (2-\overline\phi(\tau))\alpha -1)}
\]
by the implicit function theorem. The numerator is strictly positive since $\overline\phi(\tau)\in[0,1)$ and $\alpha<1$. 
The denominator is strictly negative: if $\alpha\leq 1/2$, then 
\[
(2-\overline\phi(\tau))\alpha -1\leq (2-\overline\phi(\tau))/2 -1=-\overline\phi(\tau))/2 <0,
\]
and if $\alpha>1/2$, then $(2\alpha-1)/\alpha< \overline\phi(\tau)$ by Lemma \ref{def of gamma}, and 
\[
(2-\overline\phi(\tau))\alpha -1< (2-(2\alpha-1)/\alpha)\alpha -1=0.
\]
Therefore, $\overline\phi'(\tau)<0$. 
\end{proof}

\section{Proofs for Section \ref{The definition of optimality}}

\begin{proof}[Proof of Lemma \ref{welfare lemma}]
Suppose that $\tau=f(\gamma)$. By Proposition \ref{proposition 1} and the law of total variance, 
\begin{align*}
\var[A]&=\E[\vart[A]]+\var[\Et[A]]=\frac{\lambda \gamma^2}{2(1-\gamma)}+\frac{\beta^2}{(1-\alpha)^2}\var[\tilde\theta],\\
\var[a_i]&=\E[\vart[a_i]]+\var[\Et[a_i]]=\frac{\lambda \gamma}{2(1-\gamma)}+\frac{\beta^2}{(1-\alpha)^2}\var[\tilde\theta].
\end{align*}
In addition, 
\[
\var[\tilde\theta]=\tau_\theta^{-1}-\tau^{-1}=
\tau_\theta^{-1}-
\frac{\lambda}{2 \beta ^2}\cdot\frac{(1-\alpha  \gamma)^2}{1-\gamma}
\]
follows from \eqref{main result eq 2}.  
The above three equations imply $W(\tau,\gamma)=W_+(\gamma)$.

Suppose that $\tau>f(0)$ and $\gamma=0$. Then, $a_i=A=\beta\tilde\theta/(1-\alpha)$ by Proposition~\ref{proposition 1}, and thus 
$\var[a_i]=\var[A] =\beta^2\var[\tilde\theta]/(1-\alpha)^2$, which implies $W(\tau,\gamma)=W_0(\tau)$.
\end{proof}

\section{Proofs for Section \ref{The optimal precision}}\label{Proofs for Section The optimal precision}

\begin{proof}[Proof of Proposition \ref{main proposition opt}]
By direct calculation, we have 
\begin{equation}
\overline{W}_+(\tau_+^*)-W_0(\infty)=
\begin{cases}
\chi& \text{ if }\gamma_+^*>0,\\
\displaystyle -\frac{\eta }{(1-\alpha )^2} & \text{ if }\gamma_+^*=0.
\end{cases}
\label{chi condition}
\end{equation}
Thus, when $\gamma_+^*=0$, 
\[
\tau^*
\begin{cases}
	=\infty & \text{ if $\eta>0$},\\
	=\tau_+^* & \text{ if $\eta<0$},\\
	\in[\tau_+^*,\infty] & \text{ if }\eta=0.
\end{cases}\]
Similarly, when $\gamma_+^*>0$, 
\[
\tau^*
\begin{cases}
	=\infty & \text{ if $\chi<0$},\\
	=\tau_+^* & \text{ if $\chi>0$},\\
	\in \{\tau_+^*,\infty\} & \text{ if $\chi=0$}.
\end{cases}\]
This implies the following.
\begin{enumerate}[(a)]
	\item If $\gamma_+^*=0$ and $\eta>0$ or if $\gamma_+^*>0$ and $\chi<0$, then $\tau^*=\infty$. 
	\item If $\gamma_+^*=0$ and $\eta<0$ or if $\gamma_+^*>0$ and $\chi>0$, then $\tau^*=\tau_+^*$.
\end{enumerate}
We show that the above conditions are equivalent to those in the proposition. 

Consider (a). It is clear that the condition in (a) holds if $\eta>0$ and $\chi<0$. 
To show the converse, assume first that $\gamma_+^*=0$ and $\eta>0$. 
Recall that $\gamma_+^*=0$ implies $\zeta -(1-2 \alpha) \eta /(1-\alpha)^2\leq 1$. 
Since 
\begin{equation}
2/(1-\alpha)-
(1-2 \alpha) /(1-\alpha)^2=1/(1-\alpha)^2>0, \label{key alpha}	
\end{equation}
we have
\[
\chi= \zeta-1-2 \eta /(1-\alpha)< \zeta-1 -(1-2 \alpha) \eta /(1-\alpha)^2\leq 0.
\] 
Next, assume that $\gamma_+^*>0$ and $\chi<0$. 
Note that $W_+(\gamma_+^*)>W_0(\infty)\geq W_0(f(0))$, i.e., $\eta\geq 0$. 
This implies $\eta> 0$ because if $\eta=0$, then $\zeta -(1-2 \alpha) \eta /(1-\alpha)^2=\zeta> 1$, and $\chi= \zeta-1-\log \zeta\geq 0$, contradicting $\chi<0$.

Consider (b). It is clear that $\eta<0$ or $\chi>0$ if the condition in (b) holds. 
To show the converse, assume that $\eta<0$ or $\chi>0$ holds. 
When $\gamma_+^*=0$ and $\eta\geq 0$, 
\[
0\geq \zeta-1 -(1-2 \alpha) \eta /(1-\alpha)^2\geq 
\zeta-1-2 \eta /(1-\alpha)=\chi
\] 
by \eqref{key alpha}, which contradicts the assumption. 
Thus, $\eta<0$ must hold when $\gamma_+^*=0$.  
When $\gamma_+^*>0$, $\log(1-\gamma_+^*)^{-1}=\log(\zeta-1 -(1-2 \alpha) \eta /(1-\alpha)^2)$. 
Thus, if $\eta<0$, then 
\begin{align*}
\chi&= \zeta-1-2 \eta /(1-\alpha)-\log(1-\gamma_+^*)^{-1}\\
&> \zeta-1 -(1-2 \alpha) \eta /(1-\alpha)^2-\log(\zeta-1 -(1-2 \alpha) \eta /(1-\alpha)^2)\geq 0	
\end{align*}
by \eqref{key alpha}. 
Therefore, $\chi<0$ must hold when $\gamma_+^*>0$.  
\end{proof}

\section{Proofs for Section \ref{Cournot}}\label{Proofs for Section Cournot}

\begin{proof}[Proof of Corollary \ref{cournot proposition}]
When $\zeta=\eta=1$, $\gamma_+^*=0$ if $\alpha\leq 1/2$ and $\gamma_+^*=(2 \alpha-1)/{\alpha ^2}$ if $\alpha> 1/2$, which implies that $\chi=-{2}/({1-\alpha }) -\log (1-\gamma_+^*)^{-1}<0$. 
Because $\eta>0$, full disclosure is optimal for all $\alpha<1$.  

Consider $\overline{W}_+(\tau)$. First, suppose that $\alpha\leq 1/2$, where $\Gamma(\tau)=\{\overline{\phi}(\tau)\}$ for all $\tau\in [\tau_\theta,\bar\tau]$. 
Because $\gamma_+^*=0$, $W_+(\gamma)$ is decreasing, and thus $\overline{W}_+(\tau)=W_+(\overline{\phi}(\tau))$ is increasing because $\overline{\phi}(\tau)$ is decreasing. 
Next, suppose that $\alpha\leq 1/2$. 
Because $\gamma_+^*=(2 \alpha-1)/{\alpha ^2}>0$, $W_+(\gamma)$ is increasing if $\gamma<\gamma_+^*$. Moreover, $f(\gamma_+^*)={2 \beta ^2}/{\lambda }=f(0)$. 
Thus, for each $\tau\in (f(0),\bar\tau)$, $\Gamma(\tau)=\{\overline{\phi}(\tau),\underline{\phi}(\tau)\}$ and $\overline{W}_+(\tau)=W_+(\overline{\phi}_+(\tau))$ because $\gamma_+^*>\overline{\phi}(\tau)>\underline{\phi}(\tau)$.   
This implies that  $\overline{W}_+(\tau)=W_+(\overline{\phi}(\tau))$ is decreasing if $\tau\in (f(0),\bar\tau)$ because $\overline{\phi}(\tau)$ is decreasing. 
\end{proof}

\begin{proof}[Proof of Corollary \ref{beauty contest proposition}]
When $\zeta=1+r$ and $\eta=1-r$, $\gamma_+^*=0$ if $r\leq (3-\sqrt{5})/2$ and $\gamma_+^*={(r^2-3 r +1)}/{(r(r-2))}>0$ if $r> (3-\sqrt{5})/2$. Let $\tau_+^*=f(\gamma_+^*)$ as before. 

Suppose that $(3-\sqrt{5})/2<r\leq 1/2$.  
Then, $\overline{W}_+(\tau)={W}_+(\overline\phi(\tau))$ for all $\tau\leq \bar\tau=f(0)$, 
so $\overline{W}_+(\tau)$ is decreasing for $\tau\in (\tau_+^*,\bar\tau)$ because 
${W}_+(\gamma)$ is increasing for $\gamma< \gamma_+^*=\overline\phi(\tau_+^*)$.  

Suppose that $r> 1/2$.  
Recall that, for all $t>(2r-1)/r=(2\alpha-1)/\alpha$, $f(t)$ is decreasing and $t=\overline{\phi}(\tau)$ for $\tau=f(t)<f((2r-1)/r)=\bar\tau$. 
Note that $\gamma_+^*-(2r-1)/r={(1-r)^2}/{((2-r) r)}>0$. 
Thus, ${W}_+(\gamma)$ is increasing if $(2r-1)/r<\gamma<\gamma_+^*=\overline\phi(\tau_+^*)$. 
Consequently, for $\tau\in (\tau_+^*,\bar\tau)=(f(\gamma_+^*),f((2r-1)/r))$, $\overline{W}_+(\tau)={W}_+(\overline\phi(\tau))$ is decreasing because $(2r-1)/r<\overline\phi(\tau)<\gamma_+^*$. 
\end{proof}

\section{Flexible acquisition with the Fisher information costs}\label{Flexible information acquisition with the Fisher information costs}

In this appendix, we consider the model of flexible information acquisition with the Fisher information cost \citep{hebertwoodford2021} instead of the mutual information cost and show that Propositions \ref{proposition 1} and \ref{main proposition 0} remain valid. 
We also demonstrate that the crowding-out effect enhanced by strategic complementarity has the essentially same effect on welfare as in the case of the mutual information cost discussed in Section \ref{Cournot}. 
      
Let $x$ be a normally distributed random variable with mean $\bar x$ and variance $\sigma^2$. 
Let $y$ be a signal about $x$, which has a conditional probability density function $p(y|x)$ given $x$.  
By regarding $x$ as a parameter determining the distribution of $y$ when $x$ is given, we can define the log likelihood function $l(x|y)\equiv \log p(y|x)$.
The following value is referred to as the Fisher information:
\[
\mathcal{I}_y(x)=\int \left(\frac{d l(x|y)}{d x}\right)^2p(y|x)dy.
\]
A cost of a private signal $y$ is called the Fisher information cost if it is proportional to the expected value of $\mathcal{I}_y(x)$:
\[
C^F(y)=\frac{c}{4} \cdot \E[\mathcal{I}_y(x)]=\frac{c}{4} \int \frac{\mathcal{I}_y(x)}{\sqrt{2\pi\sigma^2}}\exp\left(-\frac{(x-\bar x)^2}{2\sigma^2}\right)\sigma dx,
\]
where $c>0$ is constant. 
The following result is a special case of Proposition 4 in \citet{hebertwoodford2021}.\footnote{This result is closely related to the well-known result that the standard normal distribution has the smallest Fisher information for location among all distributions with variance less than or equal to one. See \citet{huber2009}, for example.}
\begin{lemma}
Consider an optimization problem  
\begin{equation}
\max_y-\E[(x-y)^2]-C^F(y), \notag\label{HW lemma}
\end{equation} 
where $y$ is a random variable with an arbitrary distribution.
An optimal solution exists and satisfies the following.
\begin{itemize}
\item If $c^{1/2}/2<\sigma^2$, then $y$ is a normally distributed random variable satisfying $\E[x|y]=y$, $\E[(x-y)^2]=\var[x]-\var[y]=c^{1/2}/2$, and $C^F(y)=c^{1/2}/2-c/(4\sigma^{2})$. 
\item If $c^{1/2}/2\geq\sigma^2$, then $y=\bar x$ and $C^F(y)=0$.
\end{itemize}	
\end{lemma}

Compare this lemma with Lemma \ref{rate distortion lemma}. 
The distribution of the optimal signal in each lemma is the same when $\lambda=c^{1/2}$.  
Because the proofs of  Propositions \ref{proposition 1} and \ref{main proposition 0} rely only on the distribution of the optimal signal in Lemma \ref{rate distortion lemma},  the same proofs go through under the assumption of the Fisher information cost. That is, Propositions \ref{proposition 1} and \ref{main proposition 0} remain valid, where $\lambda$ is replaced with $c^{1/2}$. 

From now on, let $\lambda=c^{1/2}$.  
In the second-period equilibrium with $\tau=f(\gamma)$, the Fisher information cost is 
\[
\lambda/2-\lambda^2/(4\vart[\alpha A+\beta \theta])
=\lambda/2-\lambda^2\gamma/(4\vart[a_i])
=\lambda\gamma/2.
\]
By replacing the mutual information cost in Lemma \ref{welfare lemma} with the above Fisher information cost, we  obtain the expected welfare:  
\begin{align}
W^F(\tau,\gamma)&=
\begin{cases}
 \displaystyle
W_+^F(\gamma) \equiv \zeta D_+(\gamma)+\eta V_+(\gamma)-\lambda\gamma/2
& \text{ if $\tau=f(\gamma)$},
 \\	
\displaystyle W_0(\tau) \equiv
 \eta V_0(\tau) & \text{ if $\tau\geq f(0)$ and $\gamma=0$}.
\end{cases}
\notag
\end{align}
For each $\gamma\in (0,1)$, it holds that 
\begin{align}
\frac{dW_+^F(\gamma)}{d\gamma} =
 \frac{\lambda}{2}\left(\zeta -\eta\frac{1-2 \alpha }{ (1-\alpha)^2}-1\right),
\notag
\end{align}
which implies that $\gamma^F_+$ maximizes $W_+(\gamma)$ over $[0,1]$, where \begin{equation}
	\gamma^F_+
\begin{cases}
= 1 
&\text{if }\zeta -(1-2 \alpha) \eta /(1-\alpha)^2>1,\\
=0 
&\text{if }\zeta -(1-2 \alpha) \eta /(1-\alpha)^2<1,\\
\in [0,1] 
&\text{if }\zeta -(1-2 \alpha) \eta /(1-\alpha)^2=1.
\end{cases}
\notag
\end{equation}

Let $\overline{W}_+^F(\tau)\equiv\max_{\gamma\in\Gamma(\tau)}W_+^F(\gamma)$ be the maximum achievable welfare given $\tau=f(\gamma)$. 
The domain of $\overline{W}_+^F(\tau)$ is the same as that of $\overline{W}_+(\tau)$,  i.e., $[\tau_\theta,\bar\tau]$. 
For $\tau<f(0)$, $\overline{W}_+^F(\tau)=W_+^F(\overline{\phi}(\tau))$, and thus 
\begin{align}
\frac{\partial \overline{W}_+^F(\tau)}{\partial \tau}\gtrless 0&\ \Leftrightarrow\ 
\zeta -(1-2 \alpha)\eta/ (1-\alpha)^2- 1\lessgtr 0.
\notag
\label{A13}
\end{align}
As a consequence, Corollary \ref{cournot proposition} on Cournot and investment games has the following counterpart in the case of the Fisher information cost. 
\begin{corollary}\label{cournot proposition fisher}
Assume that $\zeta=\eta=1$.  
For $\tau<f(0)$,  $\overline{W}_+^F(\tau)$ is increasing if $\alpha< 1/2$ and decreasing if $\alpha>1/2$. 
\end{corollary}
That is, more precise public information can reduce the total profit if and only if the game exhibits strong strategic complementarity, which is essentially the same as in the case of the mutual information cost. 

Similarly, Corollary \ref{cournot proposition} on a beauty contest game has the following counterpart in the case of the Fisher information cost. 

\begin{corollary}\label{beauty contest proposition fisher}
Assume that $\alpha=r$ and $(\zeta,\eta)=(1+r,1-r)$. 
If $r> (3-\sqrt{5})/2$, then $\overline{W}_+(\tau)$ is deceasing for $\tau<f(0)$. 
\end{corollary}

That is, welfare can decrease with public information even if $r<1/2$, which is also essentially the same as in the case of the mutual information cost.  

We can also identify the optimal precision. 
No disclosure can be optimal, which is different from the case of the mutual information cost. 
The difference arises from the following fact: the mutual information cost is proportional to $\log(1-\gamma)^{-1}$ and unbounded, whereas the Fisher information cost is proportional to $\gamma$ and bounded. 
That is, the cost reduction cannot be substantial enough to make no disclosure suboptimal in the case of the Fisher information cost.

\begin{proposition}\label{Proposition Fisher}
Assume that the agents have the Fisher information cost. 
Let $\bar\gamma=\overline\phi(\tau_\theta)$. 
The optimal precision $\tau^*$ is uniquely given by
\[
\tau^*=
\begin{cases}
	\infty & \text{ if $\eta>0$ and $\zeta< ({(1-2 \alpha ) \bar\gamma+1})\eta/({(1-\alpha )^2 \bar\gamma})+1$},\\
	0 & \text{ if $\eta\leq 0$ and $\zeta>({(1-2 \alpha ) \bar\gamma+1})\eta/({(1-\alpha )^2 \bar\gamma})+1$},\\
	0 & \text{ if $\eta<0$ and $\zeta >(1-2 \alpha) \eta /(1-\alpha)^2+1$},\\
	f(0) & \text{ if $\eta<0$ and $\zeta <(1-2 \alpha) \eta /(1-\alpha)^2+1$}.
\end{cases}
\]
\end{proposition}
\begin{proof}
Suppose that $\eta>0$. The optimal precision is either $\infty$ or $\tau_\theta$ because $W_0(\tau)$ is strictly increasing. By direct calculation, 
\[
W_0(\infty)-W_+^F(\overline\phi(\tau_\theta))=W_0(\infty)-W_+^F(\bar\gamma)=
\frac{\lambda}{2}
\left(\frac{ 1+\bar\gamma(1-2 \alpha)}{(1-\alpha)^2}\eta -\bar\gamma \zeta+\bar\gamma\right),
\]	
which is strictly positive if and only if $\zeta< ({(1-2 \alpha ) \bar\gamma+1})\eta/({(1-\alpha )^2 \bar\gamma})+1$.  

Suppose that $\eta<0$. The optimal precision is either $f(0)$ or $\tau_\theta$ because $W_0(\tau)$ is strictly decreasing, and 
\[
W_0(f(0))-W_+^F(\overline\phi(\tau_\theta))=W_+^F(0)-W_+^F(\bar\gamma)
=
- \frac{\lambda\bar\gamma}{2}\left(\zeta -\eta\frac{1-2 \alpha }{ (1-\alpha)^2}-1\right),
\]	
which is strictly positive if and only if $\zeta <(1-2 \alpha) \eta /(1-\alpha)^2+1$. 
\end{proof}

\section{Rigid acquisition with linear information costs}\label{Rigid information acquisition with linear information costs}
In this appendix, we compare the model of rigid information acquisition with linear information costs \citep{colombofemminis2008,colombofemminis2014,ui2022} with the model in this paper. 
We focus on linear information costs (linear in the precision of private information) because the crowding-out effect is largest among all convex information costs \citep{ui2014}. 
We show that the crowding-out effect is larger in the case of flexible information acquisition than in the case of rigid information acquisition if and only if the game exhibits strategic complementarity. 

An agent receives a private signal $\tilde\theta_i$ that follows a normal distribution in the second period at a cost $c\tau_i$, where $c>0$ is constant and $\tau_i\geq 0$ is the the precision of private information defined as $\tau_i\equiv 1/\var[\theta|\tilde\theta_i]-\tau_\theta$. 
A private signal is conditionally independent across agents given $\theta$ and $\tilde\theta$. 
The second period subgame has a unique equilibrium \citep{colombofemminis2014}, and the precision of private information is given by\footnote{See \citet{colombofemminis2014} and \citet{ui2022}.}
\[
\psi_c(\tau)\equiv
\begin{cases}
 \left({\beta}/{\sqrt{c}}-\tau\right)/(1-\alpha) &\text{ if }\tau< \beta/\sqrt{c},\\
0 &\text{ if }\tau\geq  \beta/\sqrt{c}.
\end{cases}
\]
For $\tau< \beta/\sqrt{c}$, $\psi_c(\tau)$ is linear and decreasing in $\tau$, and $\psi_c'(\tau)=-1/(1-\alpha)$ is decreasing in $\alpha$; that is, the crowding-out effect is more siginificant when $\alpha$ is larger, which is essentially the same as in the case of flexible information acquisition. 

The total amount of information also has the same property as in the case of flexible information acquisition. 
The mutual information of an agent's signals $(\tilde\theta,\tilde\theta_i)$ and a state $\theta$ is calculated as 
\begin{equation}
\I(\tilde\theta,\tilde\theta_i;\theta)=\frac{1}{2}\log\frac{\big|\var[\tilde\theta,\tilde\theta_i]\big|\,\big|\var[\theta]\big|}{\big|\var[\tilde\theta,\tilde\theta_i,\theta]\big|}=
I_c(\tau)\equiv \frac{1}{2}\log\frac{\tau+\psi_c(\tau)}{\tau_\theta},
\label{MI linear}	
\end{equation}
where $\big|\var[x_1,\ldots,x_n]\big|$ denotes the determinant of the covariance matrix of a random vector $(x_1,\ldots,x_n)$. 
Then, we obtain the following counterpart for Proposition \ref{main proposition 0}.
\begin{proposition}\label{crowding out linear 1}
If $\tau<\beta/\sqrt{c}$, then 
\begin{equation}
\frac{d\I(\tilde\theta,\tilde\theta_i;\theta)}{d\tau}=\frac{d{{I}_c}(\tau)}{d\tau}=-\frac{\alpha}{2(1-\alpha)(\tau+\psi_c(\tau))}. \label{main prop 1 linear eq}
\end{equation}
Thus, ${I}_c$ 
is increasing in $\tau$ if $\alpha<0$, decreasing in $\tau$ if $\alpha>0$, and independent of $\tau$ if $\alpha=0$. 	
\end{proposition}

We compare $dI_{\overline\phi}(\tau)/d\tau$ and $dI_{c}(\tau)/d\tau$ when $I_{\overline\phi}(\tau)=I_c(\tau)$, i.e., the total amount of information is the same for both types of information acquisition.  
If $dI_{\overline\phi}(\tau)/d\tau<dI_{c}(\tau)/d\tau$, the change in the total amount of information caused by a one-unit increase in public information is smaller under flexible information acquisition, which means that the crowding-out effect is larger.  
This is true if and only if the game exhibits strategic complementarity, as shown by the next proposition.

\begin{proposition}\label{crowding out linear 2}
Suppose that ${I}_{\overline\phi}(\tau)={I}_c(\tau)$ and $\tau<\min\{f(\max\{0,(2\alpha-1)/\alpha\}),\beta/\sqrt{c}\}$. 
Then, 
\begin{equation}
dI_{\overline\phi}(\tau)/d\tau\lessgtr dI_{c}(\tau)/d\tau\ \Leftrightarrow\ \alpha\gtrless 0.
\label{linear vs flexible}	
\end{equation}
\end{proposition}
\begin{proof}
By \eqref{MI linear} and\eqref{total amount of information 2}, 
\[
{I}_{\overline\phi}(\tau)-{I}_c(\tau)=
\frac{1}{2}\left(\log\frac{\tau}{1-\overline\phi(\tau)}
-\log(\tau+\psi_c(\tau))\right)=0,\]	
and thus $\tau+\psi_c(\tau)={\tau}/({1-\overline\phi(\tau)})$. 
Therefore, 
\begin{align*}
\frac{d{{I}_{\overline\phi}}(\tau)}{d\tau}-\frac{d{{I}_c}(\tau)}{d\tau}
&=\frac{\alpha\overline\phi'(\tau)}{1-\alpha\overline\phi(\tau)}+\frac{\alpha}{2(1-\alpha)(\tau+\psi_c(\tau))}	\\
&=\alpha\left(\frac{\overline\phi'(\tau)}{1-\alpha\overline\phi(\tau)}+\frac{1-\overline\phi(\tau)}{2(1-\alpha)\tau}	\right)\\
&=\alpha\left(\frac{1/f'(\gamma)}{1-\alpha\gamma}+\frac{1-\gamma}{2(1-\alpha)f(\gamma)}	\right)\\
&=\frac{\lambda \alpha (\alpha  \gamma -1)^3}{4 (1-\alpha) \beta ^2 (1-  (2-\gamma)\alpha)},
\end{align*}
where $\gamma=\overline\phi(\tau)$. 
Because $1-(2-\gamma)\alpha>0$ (see the proof of Lemma \ref{lemma: crowding-out}), 
we obtain \eqref{linear vs flexible}. 
\end{proof}

In a game with strategic complementarity, 
a one-unit increase in public information causes a greater decrease in the total amount of information (i.e., the crowding-out effect is larger) under flexible information acquisition. 
In a game with strategic complementarity, 
it causes a greater increase in the total amount of information (i.e., the crowding-out effect is smaller)  under flexible information acquisition. 
This implies that the crowding-out effect is more sensitive to the slope of the best response in the flexible information acquisition model than in the rigid information acquisition model. 





\section{The crowding-in effect of public information}\label{total information under multiple equilibria}

In this appendix, we study the problem of Section \ref{The crowding-out effect of public information} in the case of multiple equilibria. 
Assume that $\alpha>1/2$ and $f(0)< \tau< f((2\alpha-1)/\alpha)$. 
By Lemma \ref{def of gamma}, $\Gamma(\tau)=\{\underline\phi(\tau),\overline\phi(\tau)\}$. 
We focus on the equilibrium with $\gamma=\underline\phi(\tau)$ and show that when more precise public information is provided, 
the agents acquire more precise private information, so the total amount of information increases. 
In contrast, in the equilibrium with $\gamma=\overline\phi(\tau)$, it is straightforward to show that every result in Section \ref{The crowding-out effect of public information} remains valid.


As discussed in Section \ref{The crowding-out effect of public information},  
more precise public information decreases private information in the equilibrium with $\gamma=\overline\phi(\tau)$.  
However, in the equilbrium with $\gamma=\underline\phi(\tau)$, more precise public information increases private information. 
This effect of public information is referred to as the crowding-in effect. 

\begin{lemma}\label{lemma: crowding-in}
If $\alpha>0$ and $f(0)< \tau< f((2\alpha-1)/\alpha)$, then $\underline\phi'(\tau)>0$ and $\overline\phi'(\tau)<0$.
\end{lemma}
\begin{proof}
The proof for $\overline\phi'(\tau)<0$ is given in the proof of Lemma \ref{lemma: crowding-out}. 
We prove $\underline\phi'(\tau)<0$. 
By the implicit function theorem, 
\[
\underline\phi'(\tau)=1/f'(\underline\phi(\tau))
=\frac{\lambda  (1-\alpha  \underline\phi(\tau))^3}{2 \beta ^2 ( (2-\underline\phi(\tau))\alpha -1)}.
\]
The numerator is strictly positive since $\underline\phi(\tau)\in[0,1)$ and $\alpha<1$. 
The denominator is strictly positive because $\underline\phi(\tau)<(2\alpha-1)/\alpha$ by Lemma \ref{def of gamma}, and 
$
(2-\overline\phi(\tau))\alpha -1> (2-(2\alpha-1)/\alpha)\alpha -1=0
$. 
Therefore, we have $\underline\phi'(\tau)<0$. 
\end{proof}

To see why the crowding-in effect arises when $\gamma=\underline\phi(\tau)$, imagine that the policymaker provides more precise public information and, at the same time, the opponents acquire more precise private information.
Then, private information is more valuable to an agent because both private information and public information are strategic complements in a game with strategic complementarity \citep{hellwigveldkamp2009}. 
Because an agent acquires a small amount of information when $\gamma=\underline\phi(\tau)$, the marginal cost with respect to the information fraction, $(1-\gamma)^{-1}/2$, is also small. 
Thus, an agent has a strong incentive to acquire more precise private information, making the crowding-in effect consistent with equilibrium.

Due to the crowding-in effect, more public information increases the total amount of information. 
More formally, the total amount of information when $\gamma=\underline\phi(\tau)$ is given by  
\begin{equation}
\I(\tilde\theta,a_i;\theta)=I_{\underline\phi(\tau)}(\tau)\equiv 
\frac{1}{2}\left(\log\frac{\tau}{1-\underline\phi(\tau)}-\log\tau_\theta\right)
=\frac{1}{2}\left(\log\frac{2 \beta ^2}{\lambda(1-\alpha  \underline\phi(\tau))^2}-\log\tau_\theta\right),\notag
\label{total amount of information 2'}
\end{equation}
and we have 
${d I_{\underline\phi(\tau)}/(\tau)}{d\tau}={\alpha\underline\phi'(\tau)}/({1-\alpha\underline\phi(\tau)})>0$. 


\end{appendices}

\end{document}